%% file: main.tex
\documentclass[11pt,letterpaper,english]{article}
\usepackage{xcolor}
\definecolor{citecolor}{RGB}{0,0,0}
\usepackage{multirow}
\usepackage{rotating}
\usepackage{array}
\usepackage[colorlinks=true, linkcolor=blue, urlcolor=blue, citecolor=blue]{hyperref}
\usepackage{authblk}

\usepackage{amssymb,amsmath,amsfonts,eurosym,geometry,ulem,graphicx,caption,color,setspace,sectsty,comment,footmisc,pdflscape,array,hyperref,subcaption,amsthm,rotating,times,hyperref}

\usepackage{algorithm}
\usepackage{algpseudocode}
\usepackage[toc,page]{appendix}

\usepackage{hyperref}

\usepackage[sectionbib,round]{natbib}
\newcommand{\parencite}[1]{\citep{#1}}
\newcommand{\textcite}[1]{\citet{#1}}
\let\cite\citet

\newtheorem{theorem}{Theorem}
\newtheorem{corollary}[theorem]{Corollary}

\newtheorem{lemma}{Lemma}
\theoremstyle{remark}

\newtheorem{assumption}{Assumption}
\theoremstyle{definition}
\newtheorem{definition}{Definition}

\usepackage{thmtools}
\usepackage{thm-restate}

\newcolumntype{L}[1]{>{\raggedright\let\newline\\arraybackslash\hspace{0pt}}m{#1}}
\newcolumntype{C}[1]{>{\centering\let\newline\\arraybackslash\hspace{0pt}}m{#1}}
\newcolumntype{R}[1]{>{\raggedleft\let\newline\\arraybackslash\hspace{0pt}}m{#1}}

\usepackage{authblk}

\newcommand{\squishlist}{
   \begin{list}{$\bullet$}
    { \setlength{\itemsep}{0pt} \setlength{\parsep}{1pt}
      \setlength{\topsep}{1pt} \setlength{\partopsep}{1pt}
      \setlength{\leftmargin}{1.5em} \setlength{\labelwidth}{1em}
      \setlength{\labelsep}{0.5em} } }

\newcommand{\squishlisttwo}{
   \begin{list}{$\bullet$}
    { \setlength{\itemsep}{0pt} \setlength{\parsep}{0pt}
      \setlength{\topsep}{0pt} \setlength{\partopsep}{0pt}
      \setlength{\leftmargin}{1em} \setlength{\labelwidth}{1.5em}
      \setlength{\labelsep}{0.5em} } }

\newcommand{\squishend}{
    \end{list}  }

\input{header}
\begin{document}

\title{Choice Models and Permutation Invariance: Demand Estimation in Differentiated Products Markets}
\author[1]{Amandeep Singh\thanks{Email: \texttt{amdeep@uw.edu}}}
\author[1]{Ye Liu\thanks{Email: \texttt{yeliu@uw.edu}}}
\author[1]{Hema Yoganarasimhan\thanks{Email: \texttt{hemay@uw.edu}}}
\affil[1]{University of Washington}
\maketitle
\begin{abstract}
Choice modeling is at the core of understanding how changes to the competitive landscape affect consumer choices and reshape market equilibria. In this paper, we propose a fundamental characterization of choice functions that encompasses a wide variety of extant choice models. We demonstrate how non-parametric estimators like neural nets can easily approximate such functionals and overcome the curse of dimensionality that is inherent in the non-parametric estimation of choice functions. We demonstrate through extensive simulations that our proposed functionals can flexibly capture underlying consumer behavior in a completely data-driven fashion and outperform traditional parametric models. As demand settings often exhibit endogenous features, we extend our framework to incorporate estimation under endogenous features. Further, we also describe a formal inference procedure to construct valid confidence intervals on objects of interest like price elasticity. Finally, to assess the practical applicability of our estimator, we utilize a real-world dataset from \cite{berry1995automobile}. Our empirical analysis confirms that the estimator generates realistic and comparable own- and cross-price elasticities that are consistent with the observations reported in the existing literature.

\bigskip

\noindent\textbf{Keywords:} Choice Models, Demand Estimation,  Permutation Invariance, Set Functions, Neural networks

\end{abstract}

\newpage
\setcounter{page}{1}
\thispagestyle{empty}

\section{Introduction}
Demand estimation is a critical component in fields of marketing, operations, and economics, enabling practitioners to model consumer choice behavior and understand how consumers react to changes in a market. This understanding helps policymakers and businesses to make informed decisions, whether it be about launching new products, adjusting pricing strategies, or analyzing the consequences of mergers \citep{nevo2000mergers, petrin2002quantifying, nevo2003new, wollmann2018trucks}. Over the years, various approaches, both parametric and non-parametric, have been developed to address the complexities inherent in demand estimation. 

Parametric methods, based on logit or probit assumptions, often require strong assumptions about the underlying choice process, limiting their ability to capture the true complexity of consumer preferences. As such, the substantive and policy (counterfactual) implications from such models could be largely biased and/or misleading. Nevertheless, they have remained popular because of a few reasons -- (1) simplicity, interpretability, and scalability to a large product set, (2) the ability to model counterfactual demand in situations involving new product introductions or mergers, and (3) the ability to handle endogenous product features. 


Non-parametric methods, on the other hand, offer a more flexible approach to demand estimation, allowing for more nuanced representations of consumer preferences without restrictive assumptions about the underlying distributions of observables and unobservables. However, despite their potential advantages, non-parametric approaches face significant challenges that have limited their widespread adoption in practice. Firstly, a common limitation across all non-parametric demand estimation methods is the ``curse of dimensionality," where the computational complexity of estimating choice functions increases exponentially with the number of products. Additionally, a large subset of these models, specifically those that are `black box' in nature and lack additional economic structure, are unable to perform counterfactual predictions, which are crucial for important tasks such as policy simulations. Furthermore, most non-parametric methods cannot properly handle endogenous product features. This inability to account for endogeneity and scalability constraints, along with the drawback of limited counterfactual analysis, significantly undermines their utility in practice, where such capabilities are often the primary objective of demand estimation.

In this paper, we make significant strides in bridging the gap between the flexibility of non-parametric methods and the tractability of parametric models. We achieve this by introducing a fundamental characterization of choice models. Our work begins by considering a broad set of choice functions, encompassing most existing choice modeling approaches in the empirical literature. We demonstrate that most choice functions exhibit specific symmetry properties. We then leverage recent advances in computer science and mathematics literature (for instance,  \cite{han2019universal, zaheer2017deep, wagstaff2019limitations}) to characterize these functions. Our characterization allows us to build on the strengths of the non-parametric approaches (e.g., not assuming a model of consumer behavior) and overcome the challenges associated with them. First, it addresses the curse of dimensionality in choice systems and enables the flexible estimation of choice functions via non-parametric estimators. 
By leveraging the inherent permutation-invariant structure of choice models, our characterization can close the estimator's generalization gap (for instance, see \citet{sannai2019improved} in the context of neural networks) by a factor of $\sqrt{J!}$ (where $J$ is the number of products).
Second, we recognize that real-world demand systems often contain unobserved demand shocks correlated with observable product features, such as prices. These shocks can lead to endogeneity issues, which can bias the estimated choice functions if not properly addressed. To tackle this challenge, we extend our framework to accommodate endogeneity. Third, our proposed approach successfully estimates counterfactual demand for scenarios that involve changes in the product set (e.g., new product introduction, mergers, product exit). Finally, we build upon recent advances in automatic debiased machine learning and provide an inference procedure for constructing valid confidence intervals on objects of interest, such as the average effect of price.

We demonstrate the effectiveness of our proposed approach using a series of numerical simulations. We consider a variety of data-generating processes in our simulations -- multinomial logit model with linear utility, random coefficients logit with linear utility, random coefficient logit with non-linear utilities, and a setting where some consumers have inattention (i.e., ignore certain products in the market). Across all these scenarios, we show that our approach can predict market shares, own-, and cross-price elasticities with relatively high accuracy, even though we do not make any assumptions about consumer behavior. 
Further, even when the underlying DGP is complex, our model can recover market shares and elasticities similar to oracle estimators (that are assumed to know the true DGP). This allows researchers and managers to use our approach in general-purpose situations without making ad-hoc assumptions about consumer behavior. 

Next, we consider counterfactual analyses and empirically show that our model can generate accurate counterfactuals in the case of new product introductions. Finally, we showcase the performance of the automatic debiasing procedure and show that we can provide consistent inference and confidence intervals over the average effect of price on demand across all the products. In sum, our extensive numerical simulations cover a wide range of DGPs and show that our approach -- (1) is able to accurately predict market shares and price elasticities for a variety of scenarios, (2) can generate realistic counterfactual predictions in cases where the product set changes, and (3) can provide inference over economic objects of interest. 

Finally, to showcase the effectiveness and applicability of our approach to real-life datasets, we use the \cite{berry1995automobile} automobile dataset and estimate the price elasticities using our non-parametric estimator (while correcting for price endogeneity). The results of our analysis align with existing literature and demonstrate the practical utility of our approach, underscoring its potential for adoption in real-world demand estimation settings. In summary, given the theoretical, practical, and computational advantages of our approach, we expect it to be easily applicable to a wide variety of demand estimation problems in both research and practice.




\section{Literature Review}
\label{sec:lit_review}

\subsection{Parametric Discrete Choice Models}
\label{ssec:discrete_choice}

Discrete choice models play a crucial role in various fields, including economics, marketing, and operations management, as they describe decision-making processes when individuals face multiple alternatives. These models typically rely on a random utility maximization assumption, which can be traced back to \citet{thurstone1927law} and \citet{marschak1959binary}.

Over time, various choice models have emerged under different specifications for the density of unobserved utility, following the general framework in \citet{marschak1959binary}. The (Multinomial) logit model, first proposed by \citet{luce1959possible}, 
is widely used for capturing systematic taste variance based on observed characteristics across alternatives. However, the logit model assumes that error terms are independent of each other and of the characteristics of the alternative. Additionally, the independent extreme value distribution results in the irrelevant alternatives (IIA) property, implying proportional substitution across alternatives.

To enable more flexible substitution patterns, Generalized Extreme Value (GEV) models were introduced, allowing for correlated unobserved utility across alternatives. The most commonly used GEV model is the nested logit model \parencite{train1987demand, forinash1993application}. In nested logit models, the unobserved error of alternatives within a nest is specified as correlated, while the marginal distribution of each unobserved error remains as a univariate extreme value.

The mixed logit model \parencite{mcfadden2000mixed}, also known as random coefficient logit (RCL), is an even more versatile model that allows for randomness in both unobserved factors and coefficients of observed characteristics. The model was first applied by \cite{boyd1980effect} and \cite{cardell1980measuring}, with random coefficients typically specified as normal or lognormal \parencite{ben1993estimation, mehndiratta1996time, revelt1998mixed} but also other distributions such as uniform and triangular \parencite{greene2003latent, train2001comparison}. When random coefficients follow a mixture distribution, the model becomes the well-known mixed logit latent class model.

Aside from logit family models (MNL, GEV, mixed logit) that have an extreme value distributed random component, the probit model \cite{hausman1978conditional} assumes a jointly normal distribution for the error term. This allows for any pattern of substitution and can handle random taste variation. \cite{blanchet2016markov} proposed a model where substitution between alternatives is a state transition in a Markov chain, which can approximate MNL, probit, and mixed logit models.

As the decision space grows, assuming that all alternatives are considered becomes unrealistic, leading to research in consideration sets and consumer search \parencite{honka2019empirical,jiang2021consumer}. These models typically impose strong parametric assumptions on the decision rule, including whether decisions are simultaneous or sequential, the stopping rule for searches, the size of the consideration set, and the functional form of the match value.

Further, real-world demand systems often contain unobserved demand shocks correlated with observable product features, such as prices. These shocks can lead to endogeneity issues, which can bias the estimated parameters if not properly addressed. Thus, various approaches have been proposed to address these issues. For instance, \cite{berry1995automobile} proposed a generalized method of moments-based estimator to estimate a random-coefficients logit model of demand using instrumental variables. Similarly, \cite{petrin2010control} demonstrated the application of control functions to resolve endogeneity in the random coefficient logit model of demand.

Two primary concerns related to discrete choice models are the need for correct model specification and the distribution of unobserved factors. To address these issues, nonparametric demand estimation models have been developed.

\subsection{Nonparametric Demand Estimation}

While parametric discrete choice models make assumptions about the distribution of unobserved variables and specify functional forms for utilities, recent research has developed more flexible semi-parametric and nonparametric approaches for demand estimation. These newer methods ease some of the restrictive assumptions and improve the computational efficiency of choice models while retaining some structure.

Early semi-parametric work \parencite{manski1987semiparametric, lewbel2000semiparametric, honore2000panel, abrevaya2000rank} focused on relaxing the distribution of random shocks in individual-level binary choice models. More recent research \parencite{khan2021inference, shi2018estimating, pakes2022moment} has extended this approach to relax the parametric assumptions on random shocks in individual-level multinomial choice models, allowing for increased flexibility, such as individual-level fixed effects. \cite{fox2016nonparametric, briesch2010nonparametric, allen2019identification, fosgerau2021identification,chitla2022nonparametric,lu2023semi,wang2023sieve} concentrate on nonparametric identification and estimation of distributions of heterogeneous unobservables, like random coefficients, in various demand models, relaxing assumptions about heterogeneity distribution. More recent work has looked at keeping the indirect utility or choice probability functions largely unspecified while retaining parametric error terms. For instance \parencite{bentz2000neural,wang2020deep,han2022neural,sifringer2020enhancing,wong2021reslogit,aouad2022representing} use neural networks to characterize the indirect utility while retaining the logit error structure. Our work, in comparison, provides a completely flexible mapping from observed product and consumer characteristics to observed demand without relying on any assumptions regarding the choice making process.

Extant research in nonparametric methods aims to completely avoid any parametric assumptions to remove any potential source of misspecification. 
\cite{berry2014identification, berry2020nonparametric} presented the identification results for nonparametric estimation of demand from market-level and individual-level data, respectively. Studies by \cite{hausman2016individual, blundell2017nonparametric, chen2018optimal} focus on individual-level data. In particular, \cite{hausman2016individual} employs a nonparametric approach to estimate consumer surplus bounds, while \cite{blundell2017nonparametric} introduces a method for consistently estimating demand functions with nonseparable unobserved taste heterogeneity, subject to the shape restriction imposed by the Slutsky inequality. \cite{chen2018optimal} concentrates on nonparametric instrumental variables and inference in individual-level data. This paper, alongside \cite{compiani2022market} and \cite{tebaldi2023nonparametric}, examines market-level data. \cite{compiani2022market} develops a nonparametric method based on Bernstein polynomials, drawing on the identification result of \cite{berry2014identification}. Meanwhile, \cite{tebaldi2023nonparametric} proposes a technique for deriving nonparametric bounds on demand counterfactuals and applies it to California's health insurance market. 

The common challenge across all extant nonparametric demand estimation work has been that the computational complexity of estimating nonparametric demand models increases exponentially with the number of products. For instance, \cite{compiani2022market} could estimate their nonparametric model for just two products. Similarly, \cite{cai2022deep} needed data on more than 100,000 markets to reasonably estimate demand with 50 products. To put this in perspective, economic datasets are much smaller; for instance, the dataset in \cite{berry1995automobile} had 150 products across only 20 markets. This curse of dimensionality has been the primary obstacle preventing the widespread adoption of nonparametric methods in demand estimation. To overcome this barrier, our work demonstrates how to exploit the inherent permutation invariant structure of choice functions to break this curse of dimensionality and flexibly estimate demand in markets with large a number of products. In addition, our work is also related to recent work in marketing \cite{wei2022estimating} exploring the use of neural networks to estimate parameters of structural models. 

Finally, our work leverages recent work in mathematics and computer science (for instance, \cite{han2019universal, zaheer2017deep, wagstaff2019limitations}), which investigate the universal approximation of symmetric and antisymmetric functions, offering fundamental characterizations for functions defined on sets. We build on this literature to characterize a general class of choice functions in demand systems.

\section{Theory}
\label{sec:theory}
\subsection{Choice Models}
In this section, we provide a general characterization of consumer choice functions. In particular, we focus on a scenario where researchers have access only to aggregate market-level demand data, while individual-level choices and characteristics remain unobserved. Aggregate demand models have been extensively studied in marketing and economics \citep{berry1995automobile, besanko1998logit, sudhir2001competitive, chintagunta2001endogeneity, albuquerque2009estimating, compiani2022market}, and are particularly useful when individual-level demand data is not available.

Suppose consumers in a market $t$ face an offer set $\mathcal{S}_t$ that can comprise any subset of $J_t$ distinct products ($\{ 1, 2, \dots, J_t\} $). We use $u_{ijt}$ to represent the index tuple $\{X_{jt}, p_{jt}, I_{it}, \varepsilon_{ijt}\}$, where $X_{jt} \in \mathbb{C}^{d}$ denotes $d$ non-price features belonging to some countable universe $\mathbb{C}^{d}$; $p_{jt} \in \mathbb{C}$ denotes the price of the product; $I_{it} \in \mathbb{C}^{l}$ denotes demographics of consumer $i$ in market $t$, we assume there are $l$ features and belong to some countable universe $\mathbb{C}^{l}$, and $\varepsilon_{ijt}$ denotes random idiosyncratic components pertinent to consumer $i$ for product $j$ in market $t$ that are not unobservable to the researcher but observable to consumers.

\begin{definition}
   [Choice Function] Given the offer set $\mathcal{S}_t \subset \{1, 2, 3, \dots, J_t\}$, we define a function $\pi: \{ u_{ijt}: j \in \mathcal{S}_t\} \rightarrow \mathbb{R}^{|\mathcal{S}_t|} $ that maps a set of index tuples $\{u_{ijt}\}_{j \in \mathcal{S}_t}$ to a $|\mathcal{S}_t|$-dimensional probability vector. Each element in the $\pi (\cdot)$ vector represents the probability of consumer $i$ choosing product $j$ in market $t$.
\end{definition}

Here we present a very general characterization of choice functions that maps the observable and unobservable components of product and individual characteristics to observed choices through some choice function $\pi$. Note that, traditionally, $u_{ijt}$ is a scalar that represents utility in choice models. However, in our framework, $u_{ijt}$ does not necessarily represent utility. Further, we have not yet imposed any assumption on $\pi$, i.e., how consumers make choices. 

We now specify a set of assumptions on the model and data-generating process below.

\begin{assumption} [Exogeneity] The unobserved error term $\varepsilon_{ijt}$ is independent and identically distributed (i.i.d.) across all products. This can be expressed as follows:
\label{assumpt_exogenous} 
\end{assumption}
$$\mathbb{P}(\varepsilon_{ijt} \mid X_{\cdot t}, p_{\cdot t}) = \mathbb{P}(\varepsilon_{ijt})$$
This assumption implies that the error term $\varepsilon_{ijt}$ is not correlated with any of the observed variables $X_{\cdot t}$ and $p_{\cdot t}$. As such, it precludes the possibility of endogenous prices and/or marketing-mix variables, as is common in observational data. We start with the basic case with exogenous covariates in this section and later in $\S$\ref{ssec:endo}, we relax this assumption and allow for endogenous covariates.

\begin{assumption} [Identity Independence] For any product $j\in \mathcal{S}_t$ and any market $t$, we assume the choice function $\pi$ does not depend on the identity of the product ($jt$). That is:
\label{assumpt_id} 
\end{assumption}
 $$\pi_{ijt}(\{u_{ikt}\}_{k \in \mathcal{S}_t}) = \pi_{ijt}(u_{ijt}, \{u_{ikt}\}_{k \in \mathcal{S}_t, k \neq j}) = \pi(u_{ijt}, \{u_{ikt}\}_{k \in \mathcal{S}_t, k \neq j})$$ 
This assumption implies two things: first, the functional form of the choice probability for different products and markets is the same; second, for any market-level heterogeneity (e.g., in the distribution $F_t(I_{it}, \varepsilon_{ijt})$), we can include them in $X_{jt}$ as features. Intuitively, this assumption suggests that conditional on product and consumer features and the unobserved error term, the choice probabilities are not functions of the identities of the products themselves. 

\begin{assumption}
    [Permutation Invariance] The choice function $\pi$ is invariant under any permutation function $\sigma_j()$ that rearranges the indices of the competitors of product $j$, such that:
\label{assumpt_perminv}
\end{assumption}
    $$\pi_{ijt} = \pi(u_{ijt}, \{u_{i\sigma_{j}(k)t}\}_{k \in \mathcal{S}_t, k \neq j})$$ 
In this assumption, we state that the choice function for product $j$ is invariant to all permutations of its competitors. This implies that the individual's choice for product $j$ is not affected by the order or identity of the other products in the market, and it only depends on the set of competitors' characteristics. 

Since researchers only observe aggregate data, we next define the aggregate demand function. In aggregate demand settings, individual-level choices are not observable and only aggregate demand is observable. It is often the case that the market-specific individual features are not observable and are assumed to be exogenously drawn from some distribution $\mathcal{F}(m_t)$, where $m_t$ represents the market-level characteristics. For the sake of notional simplicity, we let $m_t$ to be the same across all markets. One can easily incorporate market-specific user demographics in the choice function. Thus the demand of product $j$ in market $t$ denoted by $\pi_{jt}$ can be expressed as follows:
\begin{equation}
\label{eq:agg_demand}
    \pi_{jt} = \int  \int \pi_{ijt}(\{u_{ikt}\}_{k \in \mathcal{S}_t})  d \mathcal{F}(m_{t})  d \mathcal{G} ( \varepsilon_{ijt}),
\end{equation} 
where $\mathcal{G} ( \varepsilon_{ijt})$ denotes the CDF of unobserved errors $\varepsilon_{ijt}$.   Since $u_{ijt}$ is determined by $\{X_{jt}, p_{jt}, I_{it}, \varepsilon_{ijt}\}$ and $I_{it}, \varepsilon_{ijt}$ are integrated out in a market. Hence, we can express $\pi_{jt}$ as a function of only the observable product characteristics --  
\begin{equation}
    \pi_{jt} = g(X_{jt}, p_{jt} ,\{X_{kt}, p_{kt}\}_{k \in S, k \neq j}). 
\end{equation}

\begin{lemma}
    For any choice function that satisfies Assumption \ref{assumpt_exogenous} and \ref{assumpt_perminv}, the aggregate demand function is permutation invariant. 
\end{lemma}

This permutation invariance of the aggregate demand function exists because, under the exogeneity assumption, the aggregate demand function is simply the sum (or integral) of individual choice functions that are themselves invariant to permutation. Hence, changes to the order of competitors have no impact on the aggregated result. When the assumption of exogeneity is not satisfied, the aggregate demand function does not retain the permutation invariance, even though the individual-level choice function exhibits permutation invariance. We will return to this issue in $\S$\ref{ssec:endo}.

Our Assumptions \ref{assumpt_id} (identity independence) and \ref{assumpt_perminv} (permutation invariance) are fairly standard in the choice modeling literature, although they might not always be explicitly stated as such. Table \ref{tbl:choice_models} summarizes models that satisfy these assumptions. Please see Web Appendix \ref{sec:appendix_per_inv} for detailed derivations of how these models satisfy permutation invariance.

\begin{table}[h]
\centering
\caption{Choice Models Satisfying Identity Independence and Permutation Invariance}
\label{tbl:choice_models}
\begin{tabular}{l p{7cm}}
\hline
Choice Model & Literature \\
\hline
Multinomial Logit Model & \cite{mcfadden1973conditional} \\
Mixed Logit Model & \cite{mcfadden2000mixed} \\
Nested Logit Model & \cite{train1987demand} \\
Random Coefficients Nested Logit & \cite{grigolon2014nested}\\
Generalized Extreme Value (GEV) Model & \cite{train2009discrete} \\
Probit Model & \cite{hausman1978conditional} \\
Latent Class Logit Model & \cite{kamakura1989probabilistic} \\
Markov Chain Choice Model & \cite{blanchet2016markov} \\
Customer Inattention Based Models & \cite{goeree2008limited}, \cite{turlo2023discrete}, \cite{compiani2022market} and \cite{joo2023rational} \\
Customer Search Models & \cite{mehta2003price} \\
\hline
\end{tabular}
\end{table}

\begin{restatable}{thm}{mainthm}\label{thm1}
For any offer set  $\mathcal{S}_t \subset \{1, 2, 3, \dots, J_t\}$, if a choice function $\pi:  \{ u_{ijt}: j \in \mathcal{S}_t\} \rightarrow \mathbb{R}^{|\mathcal{S}_t|}$ where  $u_{ijt}$ represents the index tuple $\{X_{jt}, p_{jt},  I_{it}, \varepsilon_{ijt}\}$ satisfies Assumption \ref{assumpt_exogenous}, \ref{assumpt_id} and \ref{assumpt_perminv}, then there exists suitable $\rho$, $\phi_1$ and $\phi_2$ such that 
$$\pi_{jt} = \rho (\phi_1 (X_{jt}, p_{jt}) + \sum_{k\neq j, k \in \mathcal{S}_t} \phi_2 (X_{kt}, p_{kt})),$$
\end{restatable}

\noindent Proof: See Web Appendix \ref{sec:appendix_proof_thrm1}.

This result is the generalization of the results shown in \cite{zaheer2017deep} and can be shown following similar arguments. The above result is very powerful and has two important takeaways: (i) The input space of the choice function does not grow with the number of products in the assortment. Rather, the input space of the choice function (i.e., $\phi_1$ and $\phi_2$) grows only as a function of the number of features of the products in consideration, and (ii) the same transformations ($\rho$, $\phi_1$, and $\phi_2$) remain valid for all offer sets, denoted by $\mathcal{S}$, irrespective of their size.  This property allows us to easily simulate the demand and entry of new products or changes in market structure, as one does with traditional parametric models. As an example, assuming $v_{jt}$ is the utility of product $j$ at market t, for the multinomial logit model one possible set of transformations could be $\phi_1(v_{jt}) = \begin{bmatrix} \exp(v_{jt})\\ 0 \end{bmatrix}$ and $\phi_2(v_{kt}) = \begin{bmatrix} 0\\ \exp(v_{kt}) \end{bmatrix}$ that generate two-dimensional vectors, and the function $\rho\left(\begin{bmatrix} \phi_1(v_{jt})\\ \sum_{k \neq j, k \in \mathcal{S}_t}\phi_2(v_{kt}) \end{bmatrix}\right) = \frac{\phi_1(v_{jt})}{\phi_1(v_{jt}) + \sum_{k \neq j, k \in \mathcal{S}_t} \phi_2(v_{kt})}$ operates on these vectors. \footnote{$\rho\left(\phi_1(v_{jt}) + \sum_{k \neq j} \phi_2(v_{kt})\right) = \rho\left(\begin{bmatrix} exp(v_{jt})\\ 0 \end{bmatrix} + \begin{bmatrix} 0 \\ \sum_{k \neq j} exp(v_{kt}) \end{bmatrix}\right) = \rho\left(\begin{bmatrix} exp(v_{jt})\\ \sum_{k \neq j} exp(v_{kt}) \end{bmatrix} \right) = \frac{exp(v_{jt})}{exp(v_{jt}) + \sum_{k \neq j} exp(v_{kt})}$} 


\subsection{Endogenous Covariates} 
\label{ssec:endo}

In this section, we relax the exogeneity assumption and handle the potential endogeneity issue that is commonplace in demand settings. Note that, when the price (or other product characteristics or market-mix variables, such as promotions, correlate with unobserved variables ($\varepsilon_{ijt}$), Assumption \ref{assumpt_exogenous} (exogeneity) is compromised. As a result, it becomes infeasible to integrate out $\varepsilon_{ijt}$ in the aggregate demand function, as we did in Equation \eqref{eq:agg_demand}. This means that the aggregate demand function loses its property of permutation invariance with respect to the observable characteristics of competitors.
To address this, we build on the approach developed in \citet{petrin2010control} to allow for endogenous observable features. Without loss of generality, we assume that price $p_{jt} \in \mathbb{C}$ is the endogenous variable and all other characteristics of the product $X_{jt} \in \mathbb{C}^{d}$ are exogenous variables. i.e., 
$$\mathbb{E}[p_{jt}\cdot \varepsilon_{ijt}] \ne 0 \quad\text{and}\quad
\mathbb{E}[X_{jt}\cdot \varepsilon_{ijt}] = 0.$$

Given valid instruments $IV_{jt}$, we can express $p_{jt}$ as 
\begin{equation}
   p_{jt}=\gamma\left(X_{jt}, IV_{jt} \right) + \mu_{jt}.
\end{equation}
At this point, no specific assumptions are made regarding the function $\gamma$. However, in the subsequent inference section, we will discuss that the estimator of $\gamma$ must be estimable at $n^{-1/2}$ in order to construct valid confidence intervals. Next, to address the issue of price endogeneity, we impose a mild restriction on the space of choice functions we consider.

\begin{assumption} [Linear Separability] The unobserved product characteristics can be expressed as the sum of an endogenous ($\text{CF}$) and exogenous component   
\begin{equation}\label{eq:res}
    \varepsilon_{ijt}=CF\left(\mu_{jt} ; \lambda\right)+\tilde{\varepsilon}_{ijt},
\end{equation}
where 
$\mathbb{E}[p_{jt}\cdot \tilde{\varepsilon}_{ijt} ] = 0$.
\label{assmpt_cf}
\end{assumption}

This assumption implies that, after controlling for $\mu_{jt}$ using the control function $CF$, the endogenous variable $p_{jt}$ is uncorrelated with the error term $\varepsilon_{ijt}$ in the model, thus it becomes exogenous. Then, we can re-write the index tuple $u_{ijt}$ as 
\begin{equation} 
    u_{ijt} = \{X_{jt},p_{jt}, CF\left(\mu_{jt};\lambda\right)+\tilde{\varepsilon}_{ijt}\},
\end{equation}
such that $\mathbb{E}\Big[\tilde{\varepsilon}_{ijt}|(X_{jt}, p_{jt}, \mu_{jt}) \Big]=0$

\begin{restatable}{thm}{endothm}\label{thm2}
For any offer set $\mathcal{S}_t \subset \{1, 2, 3, \dots, J_t\}$, if a choice function $\pi:  \{ u_{ijt}: j \in \mathcal{S}_t\} \rightarrow \mathbb{R}^{|\mathcal{S}_t|}$  where  $u_{ijt}$ represents the index tuple $\{X_{jt}, p_{jt}, I_{it}, \varepsilon_{ijt}\}$ satisfies Assumptions \ref{assumpt_id} to \ref{assmpt_cf}. Then under the condition of knowing the true  function ($\gamma_0$) of $\gamma$,  there exists suitable $\rho$, $\phi_1$ and $\phi_2$ such that 

$$\pi_{jt} = \rho (\phi_1 (X_{jt}, p_{jt}, 
 \mu_{jt}(\gamma_0)) + \sum_{k\neq j, k \in \mathcal{S}} \phi_2 (X_{kt}, p_{jt}, \mu_{kt}(\gamma_0))),$$

\end{restatable}
 
The result follows straightforwardly from the observation that after controlling for $CF(\mu_{jt}; \lambda)$ the unobservable component $\tilde{\varepsilon}$ is exogenous. This implies the aggregate demand function is invariant under any permutation applied to the competitors of product $j$. The result demonstrates that endogeneity can be addressed by using the residuals from Equation \eqref{eq:res} as an additional set of features along with observable product characteristics.

\subsection{Inference} 
\label{ssec:inference}

This paper aims to estimate choice functions flexibly using non-parametric estimators. However, often in social science contexts, researchers and managers are also interested in conducting inference over some economic objects. Note that because non-parametric regression functions are estimated at a slower rate compared to parametric regressions, it is often infeasible to construct confidence intervals directly on the estimated $\hat{\pi}$. However, it is generally possible to perform inference and construct valid confidence intervals for specific economic objects that are functions of $\pi$. In this section, we will provide an example of one such important economic object and demonstrate how to construct valid confidence intervals for it. This will be done by leveraging the recent advances in automatic debiased machine learning as shown in the works of \citet{ichimura2022influence, chernozhukov2022automatic, chernozhukov2022locally, chernozhukov2021automatic}, and others. However, unlike existing automatic debiased machine learning setups we also have to account for an additional first-stage estimator $\hat{\gamma}$.

In demand estimation, researchers are often interested in estimating the average effect of a price change on the demand for a product, as it can significantly influence market dynamics, pricing strategies, and regulatory decisions. To proceed with our analysis, let $w_{jt} = (y_{jt}, p_{jt}, X_{jt}, \{p_{kt},X_{kt}\}_{k \ne j})$ and $z_{jt} = (p_{jt}, X_{jt}, \{p_{kt},X_{kt}\}_{k \ne j})$ represent the variables associated with product $j$ in market $t$. Here, $p_{jt} \in \mathbb{C}$ denotes the observed prices, $X_{jt} \in \mathbb{C}^{d-1}$ represents other product characteristics, and $y_{jt} \in \mathbb{R}$ refers to the observed demand for product $j$ in market $t$, such as market shares or log shares. Note that either the observed price  ($p_{jt}$) or other characteristics ($X_{jt}$) could be endogenous. For simplicity and without loss of generality, we focus on $p_{jt}$ as the endogenous variable in the following analysis. 


The average effect of a price change\footnote{The expression for the average effect of a price change can be adapted to represent average price elasticity by placing the known and fixed value of $\Delta p_{jt}$ in the denominator.}  can be expressed as the difference between the demand function $\pi_{jt}(\cdot;\gamma)$ evaluated at the original price $p_{jt}$ and at the price incremented by $\Delta p_{jt}$, given by the following expression:
\begin{align*}
m(w_{jt},\pi(\cdot;\gamma)) = \pi(p_{jt} + \Delta p_{jt}, X_{jt}, \{p_{kt},X_{kt}\}_{k \ne j});\gamma) - \pi(p_{jt}, X_{jt}, \{p_{kt},X_{kt}\}_{k \ne j});\gamma).
\end{align*}

The parameter of interest, $\theta_0$, is the expected value of this price change effect over the true population distribution\footnote{We assume the data reflects the true population.} of $w_{jt}$, which can be calculated as:
\begin{align*}
\theta_0 = \mathbb{E}[m(w_{jt},\pi(\cdot;\gamma))] = \mathbb{E}[\pi(p_{jt} + \Delta p_{jt}, X_{jt}, \{X_{kt}\}_{k \ne j};\gamma) - \pi(p_{jt}, X_{jt}, \{X_{kt}\}_{k \ne j};\gamma)].
\end{align*}

In summary, the average effect of a price change on demand, denoted by $\theta_0$, is calculated by evaluating the difference between the demand function at the original price and at the price incremented by $\Delta p_{jt}$, and then computing the expected value of this difference.

In practice, we estimate $\theta_0$ by computing its empirical analog using the estimated demand function $\hat{\pi}$ and first-stage estimator $\hat{\gamma}$, i.e., 
\begin{equation}
    \hat{\theta} = \frac{1}{n}\sum_{t=1}^{n}m(w_{jt},\hat{\pi}(z_{jt};\hat{\gamma})),
\end{equation}
where $n$ is the number of observations. 
When parametric methods are employed to estimate $\hat{\pi}$ and $\hat{\gamma}$, the estimator for $\hat{\theta}$ is generally $\sqrt{n}$-consistent, assuming that the model is correctly specified. However, $\sqrt{n}$-consistency may not hold when non-parametric estimators are used, particularly if the first-order bias does not vanish at a rate of $\sqrt{n}$. Irrespective of the method used to estimate $\pi$, this is often the case, as flexible estimation of $\pi$ always requires some form of regularization and/or model selection. Debiasing techniques are required to mitigate the effects of regularization and/or model selection when learning flexible demand models. These approaches can help improve the performance of the estimator and facilitate valid inference with $\hat{\theta}$. We therefore adapt recent debiasing techniques developed in recent automatic debiased machine learning literature (see \cite{chernozhukov2022automatic}). Specifically, we will focus on problems where there exists a square-integrable random variable $\alpha_0(z)$
such that $\forall$ $||\gamma - \gamma_0||$ small enough --
\begin{align*}
     \mathbb{E}[m(w_{jt},\pi(z_{jt};\gamma))] = \mathbb{E}[\alpha_0(z_{jt})\pi(z_{jt};\gamma)] \\
              \forall \pi\text{ with }\mathbb{E}[\pi_{jt}(z_{jt};\gamma)^2] < \infty
 \end{align*}

By the Riesz representation theorem, the existence of such $\alpha_0(z_{jt})$ is equivalent to $\mathbb{E}[m(w_{jt},\pi(z_{jt};\gamma))] $ being a mean square continuous functional of $\pi(z_{jt};\gamma)$. Henceforth, we refer to $\alpha_0(z)$ as Riesz representer (or RR). \cite{newey1994asymptotic} shows that the mean square continuity of $\mathbb{E}[m(w_{jt},\pi_{jt}(z_{jt};\gamma))]$ is equivalent to the semiparametric efficiency bound of $\theta_0$ being finite. Thus, our approach focuses on regular functionals. Similar uses of the Riesz representation theorem can be found in \cite{ai2007estimation}, \cite{ackerberg2014asymptotic}, \cite{hirshberg2020debiased}, and \cite{chernozhukov2022automatic} among others.  The debiasing term in this case takes the form $\alpha(z_{jt})(y_{jt}-\pi(z_{jt};\gamma))$. To see that, consider the score $m(w_{jt},\pi(z_{jt};\gamma)) +\alpha(z_{jt})(y_{jt}-\pi(z_{jt};\gamma))-\theta_0$. It satisfies the following mixed bias property:
$$
\begin{aligned}
\mathbb{E}[m(w_{jt},\pi(z_{jt};\gamma)) & +\alpha(z_{jt})(y_{jt}-\pi(z_{jt};\gamma))-\theta_0] \\
& =-\mathbb{E}\left[\left(\alpha(z_{jt})-\alpha_0(z_{jt})\right)\left(\pi(z_{jt})-y_{jt}\right)\right] .
\end{aligned}
$$
This property implies double robustness \parencite{robins1994estimation, funk2011doubly} of the score. That is, if either $ \alpha(z_{jt}) $ is correctly estimated, which would mean $ \alpha(z_{jt}) - \alpha_0(z_{jt}) = 0 $, or $ \pi(z_{jt}) $ is correctly estimated, implying $ \pi(z_{jt}) - y_{jt} = 0 $, then the term $ (\alpha(z_{jt}) - \alpha_0(z_{jt})) (\pi(z_{jt}) - y_{jt}) $ will be zero. This results in the score going to zero, thereby making the estimator consistent for $ \theta_0 $. A debiased machine learning estimator of $\theta_0$ can be constructed from this score and first-stage learners $\widehat{\pi}$ and $\widehat{\alpha}$. Let $\mathbb{E}_n[\cdot]$ denote the empirical expectation over a sample of size $n$, i.e., $\mathbb{E}_n[x_i]=\frac{1}{n} \sum_{i=1}^n x_i$. We consider:
$$
\widehat{\theta}=\mathbb{E}_n[m(w_{jt}; \widehat{\pi})+\widehat{\alpha}(z_{jt})(y_{jt}-\widehat{\pi}(z_{jt}))] \text {. }
$$
The mixed bias property implies that the bias of this estimator will vanish at a rate equal to the product of the mean-square convergence rates of $\widehat{\alpha}$ and $\widehat{\pi}$. Therefore, in cases where the demand function $\pi$ can be estimated very well, the rate requirements on $\widehat{\alpha}$ will be less strict, and vice versa. More notably, whenever the product of the mean-square convergence rates of $\widehat{\alpha}$ and $\widehat{f}$ is larger than $\sqrt{n}$, we have that $\sqrt{n}\left(\widehat{\theta}-\theta_0\right)$ converges in distribution to centered normal law $N\left(0, \mathbb{E}\left[\psi_0(w_{jt})^2\right]\right)$, where
$$
\psi_0(w_{jt}):=m\left(w_{jt} ; \pi_0\right)+\alpha_0(z_{jt})\left(y_{jt}-\pi_0(z_{jt})\right)-\theta_0,
$$
as proven formally in Theorem 3 of \cite{chernozhukov2022automatic}. Results in \cite{newey1994asymptotic} imply that $\mathbb{E}\left[\psi_0(w_i)^2\right]$ is a semiparametric efficient variance bound for $\theta_0$, and therefore the estimator achieves this bound.

\begin{restatable}{thm}{mainthm3}\label{thm3}
[\cite{chernozhukov2021automatic}] One can view the Riesz representer as the minimizer of the loss function:
$$
\begin{aligned}
\alpha_0 & =\underset{\alpha}{\arg \min } \mathbb{E}\left[\left(\alpha(z_{jt})-\alpha_0(z_{jt})\right)^2\right] \\
& =\underset{\alpha}{\arg \min } \mathbb{E}\left[\alpha(z_{jt})^2-2 \alpha_0(z_{jt}) \alpha(z_{jt})+\alpha_0(z_{jt})^2\right] \\
& =\underset{\alpha}{\arg \min } \mathbb{E}\left[\alpha(z_{jt})^2-2 m(w_{jt} ; \alpha)\right],
\end{aligned}
$$
\end{restatable}
In our earlier discussions, we employed the moment function of $\pi$, whereas in Theorem \ref{thm3}, we focus on the moment function of $\alpha$. This shift is justified by the Riesz Representation Theorem, which implies $\mathbb{E}[m(w_{jt}; \pi)] = \mathbb{E} [\alpha_0(z_{jt}) \pi(z_{jt})]$. Given that $\pi$ can represent any function, substituting $\alpha$ for $\pi$ is permissible, thereby validating the transition from the second to the third line in Theorem \ref{thm3}. We use the above theorem to flexibly estimate the RR. The advantage of this approach is that it eliminates the need to derive an analytical form for the RR estimator, allowing it to be addressed as a simple computational optimization problem.

\begin{restatable}{thm}{mainthm4}\label{thm4}
[\cite{chernozhukov2021automatic}] Let $\delta_n$ be an upper bound on the critical radius (\cite{wainwright2019high}) of the function spaces:    
\end{restatable}
$$
\begin{gathered}
\left\{z \mapsto \zeta\left(\alpha(z)-\alpha_0(z)\right): \alpha \in \mathcal{A}_n, \zeta \in[0,1]\right\} \text { and } \\
\left\{w \mapsto \zeta\left(m(w ; \alpha)-m\left(w ; \alpha_0\right)\right): \alpha \in \mathcal{A}_n, \zeta \in[0,1]\right\}
\end{gathered}
$$
and suppose that for all $f$ in the spaces above: $\|f\|_{\infty} \leq 1$. Suppose, furthermore, that $m$ satisfies the mean-squared continuity property:
$$
\mathbb{E}\left[\left(m(w ; \alpha)-m\left(w ; \alpha^{\prime}\right)\right)^2\right] \leq M\left\|\alpha-\alpha^{\prime}\right\|_2^2
$$
for all $\alpha, \alpha^{\prime} \in \mathcal{A}_n$ and some $M \geq 1$. Then for some universal constant $C$, we have that w.p. $1-\zeta$ :
$$
\begin{aligned}
\left\|\widehat{\alpha}-\alpha_0\right\|_2^2 \leq C( & \delta_n^2 M+\frac{M \log (1 / \zeta)}{n} \\
& \left.+\inf _{\alpha_* \in \mathcal{A}_n}\left\|\alpha_*-\alpha_0\right\|_2^2\right)
\end{aligned}
$$
The critical radius has been widely studied in various function spaces, such as high-dimensional linear functions, neural networks, and superficial regression trees, often showing $\delta_n = O\left(d_n n^{-1 / 2}\right)$, where $d_n$ represents the effective dimensions of the hypothesis spaces (\cite{chernozhukov2021automatic}). In our research, we focus on applying Theorem 3 from an application standpoint to neural networks.

To that end, we make the following assumptions.
\begin{assumption}
     \label{assmp-1}
      (i) $\alpha_0(z)$ is bounded, (ii) $\forall$ $||\gamma - \gamma_0||$ small enough, $\mathbb{E}[(y-\pi_0(z_{jt}; \gamma))^2|z_{jt}]$  is bounded, and  (iii) $\mathbb{E}[m(w_{jt}, \pi_0(z_{jt};\gamma_0))^2] < \infty$.
 \end{assumption}
 These assumptions are standard regularity conditions used in the automatic machine learning literature. 
 
 \begin{assumption}
 \label{assmp-2}
     i) $\forall$ $||\gamma - \gamma_0||$ small enough $||\hat{\pi}(;\gamma) - \pi_0(;\gamma)|| \xrightarrow[]{p} 0$ and $ ||\hat{\alpha} - \alpha_0|| \xrightarrow[]{p} 0$; ii) $\sqrt{n}||\hat{\alpha} - \alpha||||(\hat{\pi}(;\gamma) - \pi_0(;\gamma)|| \xrightarrow[]{p} 0$;
iii) $\hat{\alpha}$ is bounded; (iv) $\sqrt{n}||\hat{\gamma}-\gamma_0||\xrightarrow[]{p} 0$
 \end{assumption}
 Intuitively these assumptions mean that (i) the estimator of both $\pi$ and $\alpha$ should be consistent for values of $\gamma$ in a close enough neighborhood of $\gamma_0$. Further, it requires that the product of  mean square error of $\hat{\alpha}$ and mean square error of $\pi$ should vanish at $\sqrt{n}-$ rate. This can be achieved if both these terms converge at least at $n^{-1/4}$ rate. Finally, we also assume that the first stage estimator $\hat{\gamma}$ is estimable at $n^{-1/2}$ rate. This limits the class of functions one can use to estimate $\gamma$.  
\begin{assumption}
\label{assmp-3}
    $m(w, \pi)$ is linear in $\pi$
     and there is $C>0$ such that
$$
\left|E\left[m(w, \pi)-\theta_{0}+\alpha_{0}(z) (y-\pi(z;\gamma))\right]\right| \leq C\left\|\pi-\pi_{0}\right\|^{2}
$$ 
\end{assumption}
\begin{restatable}{prop}{mainprop}\label{prop1}
     If Assumptions 5-7 are satisfied then for  $V = E[\{m(w, \pi_0(z;\gamma_0)) - \theta_0 $\\  $+ \alpha_0(z)(y-\pi_0(z;\gamma_0))\}^2]$,
     \begin{equation*}
         \sqrt{n}(\hat{\theta} - \theta_0) \xrightarrow[]{D} N(0, V ), \hat{V} \xrightarrow[]{p} V.
     \end{equation*}  
\end{restatable}

We show the proof in Web Appendix \ref{sec:appendix_proof_thrm1}. This theorem shows that if $\hat{\gamma}$ is estimable at a fast enough rate one can still construct valid confidence intervals for $\hat{\theta}$. This result can be shown following similar arguments as in \cite{chernozhukov2022locally}. Finally, we note that while the above arguments focus on the estimation of the average effect of a price change on demand, we can follow the same arguments to derive inference results for other economic quantities of interest, e.g., the effect of changing some product features on demand.

\section{Estimation Procedure}
\label{sec:estimation}
Based on theoretical results presented above, we now outline an estimation procedure for both the choice function ($\pi$) and the average effects of price changes ($\theta$). 

Consider a dataset, where $\{y_t, z_t, IV_t \}_{t=1}^{n}$ are independently and identically distributed. Here, $y_t$ is the vector of market shares in market $t$, $z_t$ is $J \times (d+1)$ matrix of product features and $IV_t$ is the $J $-dimension vector of instrumental variables. 

\begin{itemize}
\item Stage 0 (Data partition): We randomly split the observed markets into L folds such that the data $D_l := \{y_t, z_t \}_{t \in l}$, where $l$ denotes the $l^{th}$ partition. Note that all the observations for one market are always in one fold. 

\item Stage 1 (Estimate $\hat{\gamma}$): For each fold $l$, we estimate $\gamma_l$ by regressing the endogenous variable on the exogenous instruments on the left out data $D_{l}^c := \{y_t, z_t \}_{t \notin l}$. We then use the cross-fitting technique, same as \cite{chernozhukov2021automatic}, to calculate the residual $\hat{\mu}_l$ of fold $l$ with estimated $\hat{\gamma}_l$ on $D_{l}^c$.

\item Stage 2 (Estimate $\hat{\pi}$ and $\hat{\theta}$): 
\begin{itemize}
    \item Stage 2a (Estimation): In the second stage, for each fold $l$, we estimate both the choice function ($\hat{\pi}$) and the Riesz estimator ($\hat{\alpha}$) on the left out data $D_{l}^c := \{y_t, z_t \}_{t \notin l}$
    \begin{equation} \label{eq:pi_loss}
    \hat{\pi}_l=\argmin_{\pi \in \mathcal{F}} \frac{1}{\sum_{t \in D_{l}^c} J_t}\sum_{t \in D_{l}^c}\sum_{j \in J_t}[(y_{jt}-\pi(z_{jt};\hat{\gamma})) ^2] 
\end{equation}

\begin{equation}\label{eq:alpha}
    \hat{\alpha}_l = \argmin_{\alpha \in \mathcal{A}} 
\frac{1}{\sum_{t \in D_{l}^c} J_t}\sum_{t \in D_{l}^c}\sum_{j \in J_t} \left[\alpha(z_{jt})^2-2 m(w_{jt} ; \alpha)\right].
\end{equation}
Based on Theorem \ref{thm2}, instead of directly estimating the function $\pi$, we decompose the estimation into three sub-components: $\rho, \phi_1$, and $\phi_2$. Specifically, for each component of our model ($\phi_1$, $\phi_2$, and $\rho$), we use a standard 3-layer neural network\footnote{For $\phi_1$ and $\phi_2$, each of the three layers consists of 64 neurons, with the output vector also featuring 64 neurons. For $\rho$, the layers are configured with 300, 100, and 64 neurons for the first, second, and third layers, respectively.}, and this is implemented without further hyperparameter tuning. We implement ReLU activation function at each layer as it is standard in feedforward designs due to simplicity and computation efficiency in gradients. Figure \ref{fig:NN} illustrates how we pass the data to the neural networks. Specifically, we pass the focal product's characteristics (price $p_{jt}$ and other product features $X_{jt}$) and the residuals ($\hat{\mu}_{jt}$) estimated from the first stage regression to the $\phi_1$. In parallel, we pass all the products' characteristics of the other products in the same market ($p_{jt}$ and $X_{kt}$) and the corresponding residuals ($\hat{\mu}_{kt}$) to the same $\phi_2$, and then sum the output up. The output of $\phi_1$ and $\phi_2$ have the same data structure (e.g., a 64-dimension vector). Next, we pass the summation of the output of $\phi_1$ and $\phi_2$ to a third neural network $\rho$. The output of $\rho$ is a scalar which represents the market share of the focal product $jt$. 

We use the same neural network structure (with the three sub-components same as $\phi_1$, $\phi_2$, and $\rho$)  to estimate $\alpha$. The only difference is that the loss function of $\alpha$ is not based on the difference between the observed and the predicted market share as in Equation \ref{eq:pi_loss}. Instead, the loss function is based on the squared difference between $\alpha$ and the moment function of $\alpha$ as stated in Equation \ref{eq:alpha} and Theorem \ref{thm3}. 

\begin{figure}
    \centering
    \includegraphics[width=0.9\textwidth]{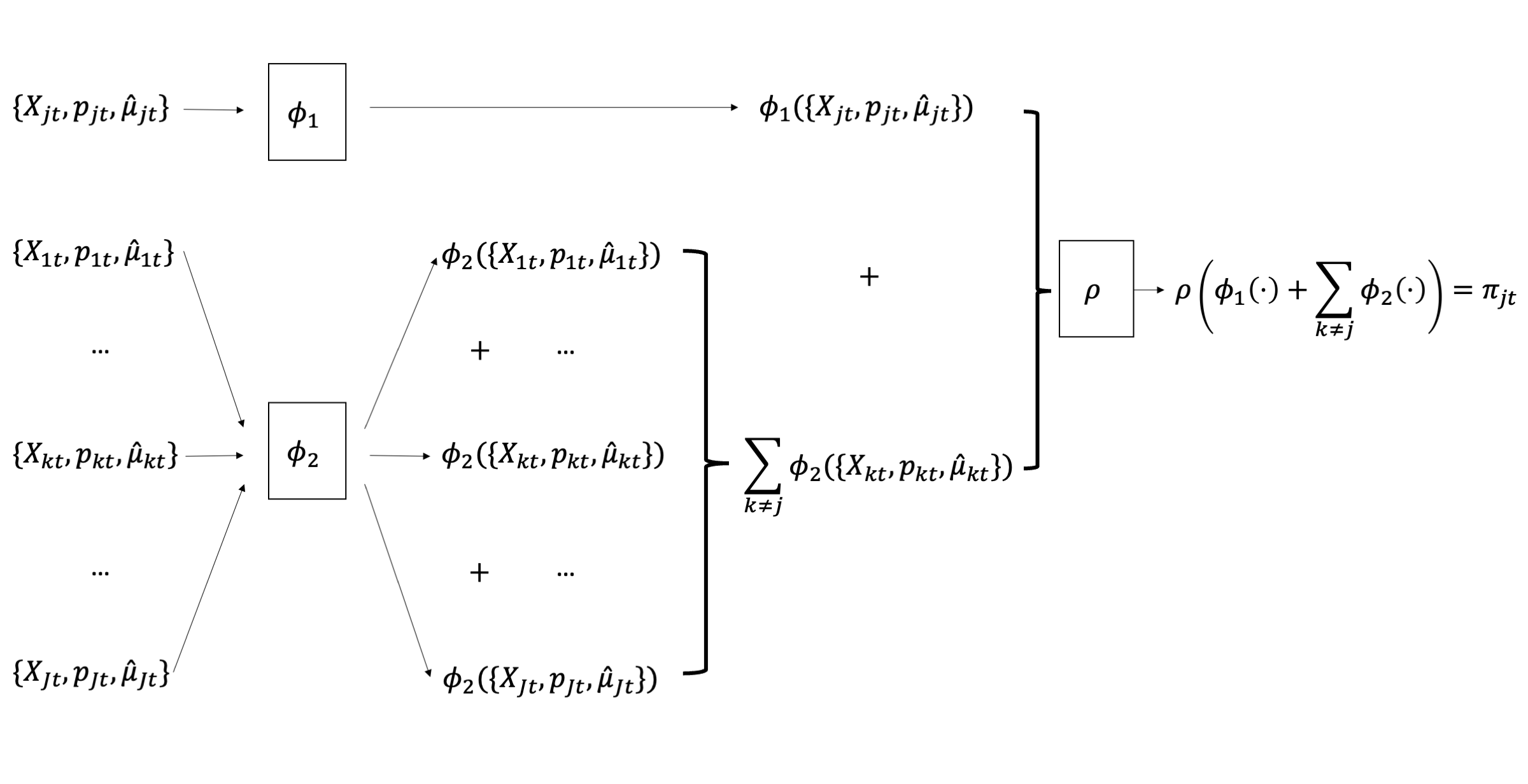}
    \caption{Illustration of Neural Network Architecture}
    \label{fig:NN}
\end{figure}

    \item Stage 2b (Cross-fitting): Now we again use the cross-fitting technique to reduce the bias when estimating $\hat{\theta}$. Specifically,  we use the estimators ($\hat{\pi}$ and $\hat{\alpha}$) estimated on  $D_{l}^{c}$ to estimate the $\hat{\theta}_l$ of $l$. By applying cross-fitting, we ensure that the nuisance functions and the parameters are estimated on separate, non-overlapping datasets.  This approach diminishes the risk of overfitting and enhances the robustness of our estimation.  And finally, to estimate $\hat{\theta}$, we randomly select one observation $t^*$ in each market $t$ and average it out across all folds. Thus the estimator for $\theta_0$ and its variance can be given as follows:
\begin{eqnarray}
\hat{\theta} &=& \frac{1}{n} \sum_{\ell=1}^{L} \sum_{t \in D_{\ell}^c}\left\{m\left(w_{t^*}, \hat{\pi}_{\ell}\right)+\hat{\alpha}_{\ell}\left(z_{t^*}\right) \left(y_{t^*}-\hat{\pi}_{\ell}(z_{t^*};\hat{\gamma})\right)\right\} \\
\hat{V} &=& \frac{1}{n} \sum_{\ell=1}^{L} \sum_{t \in D_{\ell}^c} 
\hat{\psi}_{t^* \ell}^{2}, \quad \hat{\psi}_{t^* \ell}=m\left(w_{t^*}, \hat{\pi}_{\ell}\right)-\hat{\theta}+\hat{\alpha}_{\ell}\left(z_{t^*}\right) \left(y_{t^*}-\hat{\pi}_{\ell}(z_{t^*};\hat{\gamma})\right) 
\end{eqnarray}

\end{itemize}
\end{itemize}

\section{Numerical Experiments} 
\label{sec:simulation}

We now present a series of simulation studies that establish the numerical performance of our approach. First, in $\S$\ref{ssec:predictive_performance}, we examine the predictive performance of our model on a series of models, including stylized discrete choice models with linear utilities as well as more general models that allow non-linear utilities and realistic consumer behaviors such as inattention. Next, in $\S$\ref{ssec:counterfactual}, we present numerical experiments that demonstrate our approach's ability to simulate counterfactuals. Finally, in $\S$\ref{ssec:inf_sim}, we demonstrate the applicability of the inference procedure proposed in $\S$\ref{ssec:inference}.


\subsection{Predictive Performance}
\label{ssec:predictive_performance}
In this section, we show that our approach can recover demand and elasticity estimates for a wide variety of settings without the knowledge of the underlying choice model and/or making parametric assumptions on consumers' behaviors. Before describing the simuations, we first describe the metrics used for comparing the performance of different estimators and the benchmark estimation approaches used.

First, to assess the predictive performance of our approach, we focus on three quantities of interest: 
\squishlist
\item Market share  ($\hat{\pi}_{jt}$)
\item Own price elasticity ($\frac{\partial \hat{\pi}_{jt} / \pi_{jt}  }{\partial p_{jt}/p_{jt}}$)
\item Cross-price elasticity ($\frac{\partial \hat{\pi}_{jt} / \pi_{jt}  }{\partial p_{k\neq j t}/  p_{k\neq j t} }$)
\squishend

In each simulation, we compare the performance of our model with the predictive performance of four baseline models:
\squishlist
    \item Multinomial Logit model (MNL)
    \item Random Coefficient Logit model (RCL)
    \item A standard neural network-based non-parametric method (NP).
    
    Here we simply use all the features of all the products in the market as input and predict a market share vector without using the permutation invariance property or correcting for endogeneity. This is similar to the approaches used by \cite{gabel2022product} and \citet{cai2022deep}, which simply use a large neural network for demand prediction. For this baseline, we tune the hyperparameters of the neural network, including the number of layers, number of nodes in each layer, learning rate, and the number of epochs using 5-fold cross-validation for each data generation. We detail the space of hyperparameters in Web Appendix \ref{appendix_hyper}. We also apply the ReLU activation for each layer. 

    \item A Mean Predictor (MP).
    
    Here we predict a uniform market share for all observed products within a market, excluding the outside option. This predicted share is set to the average of all observed market shares in that market\footnote{For instance, consider a market with three observed products having market shares of 0.2, 0.3, and 0.4, respectively. This implies the outside option holds a market share of 0.1. In such a scenario, the MP would predict the market share for each of the three products to be 0.3, which is the average ((0.2 + 0.3 + 0.4) / 3).}. 
    This serves as a baseline estimator as it does not account for individual product characteristics, providing a benchmark for the simplest prediction scenario. Additionally, comparing the performance of other models with MP also enables us to quantify the natural decrease in the MAE and RMSE when the number of products increases because the magnitude of the market shares decreases as the number of products increases\footnote{For instance, in a market with 100 products, the MAE and RMSE for any estimator are expected to be quite low. Comparing with the MP in such situations helps establish a baseline for MAE or RMSE.}. 
\squishend

\subsubsection{Multinomial Logit and Random Coefficient Logit with Linear Utility}
\label{sssec:mnl_rcl}


We first consider the two standard Data Generating Processes (DGPs) used in the demand estimation literature that use linear utility-based choice models  -- (1) Multinomial Logit model, and (2) the Random Coefficient Logit model. For both cases, we consider a setting with 10 products ($J = 10$), 100 markets ($T = 100$), one price feature, and 10 non-price features ($d = 10$). We define the utility $u_{ijt}$ that consumer $i$ in market $t$ derives from product $j$ as the following linear function:
\begin{equation}
u_{ijt} = \alpha_i p_{jt} + \beta_{i} X_{jt} + \varepsilon_{ijt},
\label{eq:linutil_rcl}
\end{equation}
where $\varepsilon_{ijt}$ represents an independently and identically distributed (iid) Type-I extreme value across products and consumers. $X_{jt} \in \mathbb{C}^{d}$ denotes the non-price features of the product. $\alpha_i, \beta_{i}$ are the model coefficients, which are kept constant for all consumers in the MNL, while in the RCL, they are normally distributed across consumers. The probability distribution of features and coefficients used are shown in Web Appendix \ref{sec:appendix_dist}. Also, the mean utility from the outside option is normalized to 0.  

We denote the market share of product $j$ in market $t$ generated from MNL by $\pi_{jt}^{MNL}$ and the market share generated from RCL by $ \pi_{jt}^{RCL}$. For each market, we generate the market shares of each product by simulating $N = 10,000$ individual choices and aggregating by each market as shown below. 
\begin{equation}
    \pi_{jt}^{MNL} = \frac{1}{N}\sum_i^{N} 1(u_{ijt} = \max_{k\in \mathcal{S}_t}(u_{ikt})) 
\end{equation} 

\begin{equation}
    \pi_{jt}^{RCL} = \frac{1}{N} \sum_i^{N} \frac{exp(\alpha_i p_{jt} +  \beta_{i} X_{jt})}{1 + \sum_{k \in \mathcal{S}_t}  exp(\alpha_i p_{kt} +  \beta_{i} X_{kt})}
    \label{eq:rcl_mktshare}
\end{equation} 
Note that for MNL, instead of simulating each individual's choice probability, we simulate each individual's choice based on the utility maximization principle. This approach ensures that when we use MNL (true model) for estimation, it does not reproduce the data perfectly.



\begin{table}[htp!]
\caption{Baseline Predictive Performance: Market Shares, Own-Elasticity, and Cross-Elasticity}
\label{tbl:baseline_summary}

\centering

\subcaption{Market Shares ($\hat{\pi}_{jt}$)}
\label{tbl:base_err}
\scalebox{0.7}{
\begin{tabular}{cccccccccccccccc}
\hline
\hline
    \# & True Model & J & T & d & \multicolumn{2}{c}{Our Model} & \multicolumn{2}{c}{MNL} & \multicolumn{2}{c}{RCL} & \multicolumn{2}{c}{NP} & \multicolumn{2}{c}{MP} & No. Obs. \\
    &             &   &   &   & MAE & RMSE & MAE & RMSE & MAE & RMSE & MAE & RMSE & MAE & RMSE & \\
\hline
  0 & MNL        & 5 & 100 & 10 & 0.0534 & 0.0834 & 0.0078 & 0.0105 & 0.0082 & 0.0052 & 0.1269 & 0.2364 & 0.2312 & 0.2220 & 2000 \\
  1 & MNL        & 10 & 20 & 10 & 0.0585 & 0.1086 & 0.0040 & 0.0131 & 0.0089 & 0.0134 & 0.1129 & 0.2191 & 0.1365 & 0.2948 & 800 \\
  2 & MNL        & 10 & 100 & 10 & 0.0333 & 0.0591 & 0.0044 & 0.0039 & 0.0026 & 0.0053 & 0.1181 & 0.1717 & 0.1422 & 0.1503 & 4000 \\
  3 & MNL        & 10 & 200 & 10 & 0.0307 & 0.1346 & 0.0032 & 0.0102 & 0.0034 & 0.0197 & 0.1096 & 0.2170 & 0.1416 & 0.2106 & 8000 \\
  4 & MNL        & 20 & 100 & 10 & 0.0194 & 0.0765 & 0.0015 & 0.0077 & 0.0023 & 0.0068 & 0.0707 & 0.2242 & 0.0768 & 0.2201 & 8000 \\ \hline
  5 & RCL        & 5 & 100 & 10 & 0.0240 & 0.0314 & 0.0307 & 0.0382 & 0.0033 & 0.0042 & 0.0456 & 0.0583 & 0.0538 & 0.0656 & 2000 \\
  6 & RCL        & 10 & 20 & 10 & 0.0206 & 0.0281 & 0.0270 & 0.0343 & 0.0034 & 0.0044 & 0.0540 & 0.0612 & 0.0418 & 0.0525 & 800 \\
  7 & RCL        & 10 & 100 & 10 & 0.0171 & 0.0231 & 0.0262 & 0.0326 & 0.0025 & 0.0033 & 0.0458 & 0.0583 & 0.0413 & 0.0514 & 4000 \\
  8 & RCL        & 10 & 200 & 10 & 0.0141 & 0.0187 & 0.0252 & 0.0318 & 0.0032 & 0.0039 & 0.0431 & 0.0559 & 0.0412 & 0.0513 & 8000 \\
  9 & RCL        & 20 & 100 & 10 & 0.0099 & 0.0140 & 0.0262 & 0.0281 & 0.0018 & 0.0024 & 0.0390 & 0.0489 & 0.0276 & 0.0354 & 8000 \\
\hline\hline
\end{tabular}}

\bigskip 

\subcaption{Own-Elasticity ($\frac{\partial \hat{\pi}_{jt} / \pi_{jt}  }{\partial p_{jt}/p_{jt}}$)}
\label{tbl:base_elas_s}
\scalebox{0.7}{
\begin{tabular}{clcccccccccccc}
\hline\hline
\# & True Model & J & T & d & \multicolumn{2}{c}{Our Model} & \multicolumn{2}{c}{MNL} & \multicolumn{2}{c}{RCL} & \multicolumn{2}{c}{NP} & No. Obs \\
& & & & & MAE & RMSE & MAE & RMSE & MAE & RMSE & MAE & RMSE & \\
\hline
0 & MNL & 5 & 100 & 10 & 0.2588 & 1.1815 & 0.1414 & 1.1232 & 0.1757 & 1.1322 & 0.4554 & 1.2115 & 8000 \\
1 & MNL & 10 & 20 & 10 & 0.3523 & 1.3557 & 0.1967 & 1.2970 & 0.2585 & 1.3118 & 1.0057 & 1.5316 & 3200 \\
2 & MNL & 10 & 100 & 10 & 0.3346 & 1.4327 & 0.2066 & 1.3851 & 0.2150 & 1.3863 & 0.9266 & 1.5857 & 16000 \\
3 & MNL & 10 & 200 & 10 & 0.3245 & 1.4131 & 0.2007 & 1.3842 & 0.2357 & 1.3876 & 0.8266 & 1.5768 & 32000 \\
4 & MNL & 20 & 100 & 10 & 0.4146 & 1.6596 & 0.3305 & 1.6203 & 0.3570 & 1.6229 & 1.0151 & 1.8353 & 32000 \\ \hline
5 & RCL & 5 & 100 & 10 & 0.1189 & 0.2310 & 0.1474 & 0.3735 & 0.0125 & 0.0305 & 0.1802 & 0.3145 & 8000 \\
6 & RCL & 10 & 20 & 10 & 0.1799 & 0.3729 & 0.2039 & 0.5326 & 0.0365 & 0.1196 & 0.3768 & 0.3254 & 3200 \\
7 & RCL & 10 & 100 & 10 & 0.1498 & 0.2862 & 0.2154 & 0.5531 & 0.0224 & 0.0685 & 0.2987 & 0.3643 & 16000 \\
8 & RCL & 10 & 200 & 10 & 0.1209 & 0.2416 & 0.2188 & 0.5512 & 0.0233 & 0.0732 & 0.2464 & 0.4241 & 32000 \\
9 & RCL & 20 & 100 & 10 & 0.1658 & 0.3533 & 1.4591 & 1.7099 & 0.0429 & 0.1319 & 0.4555 & 0.4741 & 32000\\
\hline\hline
\end{tabular}}

\bigskip 

\subcaption{Cross-Elasticity ($\frac{\partial \hat{\pi}_{jt} / \pi_{jt}  }{\partial p_{k\neq j t}/  p_{k\neq j t} }$)}
\label{tbl:base_elas_c}
\scalebox{0.7}{
\begin{tabular}{clcccccccccccc}
\hline\hline
  \# & True Model &   J &   T &   d &   \multicolumn{2}{c}{Our Model} &   \multicolumn{2}{c}{MNL} &   \multicolumn{2}{c}{RCL} & \multicolumn{2}{c}{NP} & No. Obs \\
  & &&&& MAE & RMSE & MAE & RMSE & MAE & RMSE & MAE & RMSE & \\
\hline
0 & MNL & 5 & 100 & 10 & 0.1349 & 0.8115 & 0.0107 & 0.6780 & 0.0118 & 0.6807 & 0.1968 & 0.9442 & 32000 \\
1 & MNL & 10 & 20 & 10 & 0.0649 & 0.6742 & 0.0043 & 0.5326 & 0.0059 & 0.5424 & 0.1527 & 0.7123 & 28800 \\
2 & MNL & 10 & 100 & 10 & 0.0862 & 0.6104 & 0.0043 & 0.5142 & 0.0049 & 0.5143 & 0.1885 & 0.7488 & 144000 \\
3 & MNL & 10 & 200 & 10 & 0.0901 & 0.6075 & 0.0042 & 0.5140 & 0.0047 & 0.5153 & 0.2110 & 0.7831 & 288000 \\
4 & MNL & 20 & 100 & 10 & 0.0482 & 0.4665 & 0.0015 & 0.4151 & 0.0016 & 0.4157 & 0.1270 & 0.5303 & 608000 \\ \hline
5 & RCL & 5 & 100 & 10 & 0.0293 & 0.0571 & 0.0492 & 0.0551 & 0.0030 & 0.0070 & 0.0617 & 0.0960 & 32000 \\
6 & RCL & 10 & 20 & 10 & 0.0261 & 0.0435 & 0.0332 & 0.0455 & 0.0039 & 0.0143 & 0.0972 & 0.0791 & 28800 \\
7 & RCL & 10 & 100 & 10 & 0.0257 & 0.0493 & 0.0324 & 0.0447 & 0.0028 & 0.0090 & 0.0795 & 0.1124 & 144000 \\
8 & RCL & 10 & 200 & 10 & 0.0213 & 0.0417 & 0.0353 & 0.0455 & 0.0035 & 0.0097 & 0.0745 & 0.1277 & 288000 \\
9 & RCL & 20 & 100 & 10 & 0.0220 & 0.0431 & 0.0204 & 0.0359 & 0.0022 & 0.0103 & 0.0715 & 0.0890 & 608000 \\
\hline\hline
\end{tabular}}

\smallskip
\footnotesize
\raggedright

Note: Table \ref{tbl:base_err}, \ref{tbl:base_elas_s}, \ref{tbl:base_elas_c}  present the baseline predictive performance for predicted market shares, own- and cross-elasticities of our model and four baseline models. J, T, and d represent the number of products, non-price features, and markets, respectively. NP denotes a benchmark non-parametric method, which is a standard neural network. MP denotes the mean predictor. 
The Mean Absolute Error (MAE) and Root Mean Square Error (RMSE) of predicted market shares for each scenario (i.e., true model, J, d, T) are computed using the test data from 20 iterations of data generation, while the MAE and RMSE of own- and cross-elasticities are computed based on the training data from 20 iterations of data generation. The column titled ``No. Obs." indicates the total number of observations for each metric. Specifically, the number of observations for market share ($\hat{\pi}_{jt}$) is calculated based on $T \times J \times 20\% $ (the portion of test data) $\times 20$ (the number of draws of simulations). The number of observations for own-elasticity ($\frac{\partial \hat{\pi}_{jt} / \pi_{jt}  }{\partial p_{jt}/p_{jt}}$) is calculated based on $T \times J \times 80\% $ (the portion of training data) $\times 20$ (the number of draws of simulations).  The number of observations for cross-elasticity ($\frac{\partial \hat{\pi}_{jt} / \pi_{jt}  }{\partial p_{k\neq j t}/  p_{k\neq j t} }$) is calculated based on $T \times J \times (J-1) \times 80\% $ (the portion of training data) $\times 20$ (the number of draws of simulations). 

\end{table}



For each DGP, we split the generated data into training data (80\%) and test data (20\%). We use the training data for estimation (both our model and the benchmark models described above).\footnote{In our simulations, we train both our model and NP with the log market shares ($log(\pi_{jt})$). The performance metrics reported for predicted market shares are computed based on the exponential values of the predicted log market shares, bringing these metrics back to market shares ($\pi_{jt}$). The performance metrics reported for elasticities are computed based on the relative change of the predicted log market shares ($\partial log(\pi_{jt})$) divided by the percentage change of the price ($\partial p_{jt}/p_{jt}$ for own-elasticity and $\partial p_{kt}/p_{kt}$ for cross-elasticity), which is equivalent to the elasticity calculated directly using the market share ($\frac{\partial \hat{\pi}_{jt} / \pi_{jt}  }{\partial p_{jt}/p_{jt}}$ for own-elasticity and  $\frac{\partial \hat{\pi}_{jt} / \pi_{jt}  }{\partial p_{k\neq j t}/  p_{k\neq j t} }$  for cross-elasticity).} For the predicted market share, we present all the model results and comparisons on the test data. For the predicted own- and cross-elasticities, we present all the model results and comparisons on the training data.\footnote{The reason we only use test data to report predicted market share accuracy is to demonstrate the model's predictive performance on unseen data. In contrast, we use training data to report accuracy in elasticities to mimic the real empirical setting where we use full data to estimate elasticity.}

Tables \ref{tbl:base_err} shows the Mean Absolute Error (MAE) and Root Mean Square Error (RMSE) in the predicted market share ($\hat{\pi}_{jt}$) for our approach as well as the baseline models.We see that when the true model is MNL, our model cannot beat RCL or MNL, which is as expected; but the error of our model is quite close to the true model. When the true model is RCL, our model can beat MNL consistently and the performance of our model is also close to the true model. Importantly, we find that our model consistently outperforms the benchmark Non-Parametric (NP) method in all data generation processes. This is despite extensive hyperparameter tuning for the NP method. 

There are two key reasons why our approach outperforms the standard neural network-based non-parametric method, especially as the number of products increases. First, unlike the standard neural network, our method can circumvent the curse of dimensionality that arises with the increase in the number of products. The standard neural network uses the stacked product features (d-dimension non-price features $X_{jt}$ and price $p_{jt}$) as input and has $(J \times (d +1) ) \times h_1$ parameters in the input layer, where $h_1$ denotes the size of the first hidden layer.\footnote{In this section, we consider only cases where there are no correlated unobservables. When there are potential endogeneity concerns, then we can include a residual $\mu_{jt}$ estimated from a first-stage regression, and then size of the input layer becomes $J \times (d+2)$. We consider settings with endogeneity in Web Appendix $\S$\ref{sec:appendix_endo} and in the experiments on inference in $\S$\ref{ssec:inf_sim}.}  In contrast, our model only uses product features as the input (with dimension $d+1$); see Figure \ref{fig:NN}.  Therefore, the parameters for our model do not scale with the number of products. Thus, as the number of products increases, our method is able to exploit this information to improve its performance, whereas the standard NP is unable to do so. Second, our model can leverage product-level market-share data more effectively. Note that one observation in the standard NN consists of one market, whereas one observation for our method consists of one product in a market. Hence, given data on $T$ markets, the number of samples available for the NP method is $T$, whereas the sample available for our model is $T \times J$. Together, these strengths of our approach lead to significantly better performance compared to a naive neural network. 

We observe both sources of performance improvement in the numerical simulation results in Table \ref{tbl:base_err}. As we vary the number of products (5, 10, and 20), the MAE of the predicted market shares from our model decreases monotonically (see simulation numbers 0, 2, and 4 for MNL, and 5, 7, and 9 for RCL). In contrast, the performance of the benchmark non-parametric estimator deteriorates as the number of products increases and becomes even worse than the Mean Prediction for the 20 product cases (simulation numbers 4 and 9). These findings demonstrate the overfitting problems and the curse of dimensionality issues discussed above. Indeed, this limitation has also been theoretically established by \citet{sannai2019improved}, who showed that for neural networks that do not take into account the inherent invariance structure, the generalization gap increases in proportion to the number of possible permutations, which is $\sqrt{J!}$ in our case. Further, we examine the performance of our model and the other benchmarks by varying the number of markets (20, 100, and 200) while keeping the number of products constant at 10; see simulations 1, 2, and 3 for MNL and 6, 7, and 8 for RCL. Although both our model and the NP method show improved performance with more markets, the non-parametric estimator is more adversely affected by a decrease in market numbers due to a more significant reduction in its sample size. This is particularly problematic in scenarios with one market, as the NP method becomes infeasible for estimation for only one sample. 

Finally, in Tables \ref{tbl:base_elas_s} and \ref{tbl:base_elas_c}, we show the predictive performance of own-elasticity ($\frac{\partial \hat{\pi}_{jt} / \pi_{jt}  }{\partial p_{jt}/p_{jt}}$) and cross-elasticity ($\frac{\partial \hat{\pi}_{jt} / \pi_{jt}  }{\partial p_{k\neq j t}/  p_{k\neq j t} }$), respectively. Again, we find that our model consistently outperforms the NP method in all scenarios for both own- and cross-elasticity predictions. When the true model is MNL, our model underperforms compared to RCL, given that RCL inherently captures the substitution pattern among products in MNL\footnote{In an extreme case, when the variance of random coefficients is zero, RCL is equivalent to MNL.}. When the true model is RCL, our model is the closest one to the true model. Unlike the market share predictions, the accuracy of own-elasticity predictions does not exhibit a clear monotonical improvement as the number of products increases. We observe a similar pattern even for the true model--the accuracy of own-elasticity decreases as the number of products increases. This suggests that it is not a deficiency of our model but due to the inherent complexities in estimating own-elasticities in markets with many products.  In the prediction of cross-elasticity, our model shows a monotonical improvement in accuracy with an increase in the number of products. Additionally, for both own- and cross-elasticities, as the number of markets increases, the performance of our model is better due to the increase in the sample size. 

So far, in the above simulations, we did not consider any endogenous explanatory variables. As discussed in $\S$\ref{ssec:endo}, our approach can easily account for endogeneity following the estimation steps in $\S$\ref{sec:estimation}. For interested readers, we present a set of numerical experiments with endogeneous explanatory variables in Web Appendix $\S$\ref{sec:appendix_endo}. The key takeaway from this analysis is that ignoring endogeneity can lead to significant biases in the estimates of own- and cross-price elasticities. Therefore, in the application to real data, we take care of endogeneity carefully; please see $\S$\ref{sec:empirical_blp} for further details.

\subsubsection{Random Coefficient Logit with Non-Linear Utility}
\label{sssec:rcl_nonlinear}
In the previous section, we focused on linear utility specifications and standard choice behaviors. Recent literature has highlighted that the oversight of non-linear relationships between features and utilities can introduce biases in the estimates \citep{allenby2004choice}. Conversely, non-parametric estimators are adept at capturing these non-linear patterns directly from the data. As a result, there has been a growing trend towards the adoption of non-linear utility functions. Therefore, we now focus on data generated from a random coefficient logit model with non-linear transformations applied to observable features. We consider a case with two product features -- price and a non-price feature $x$. We apply a non-linear transformation $g(x)$. Following \cite{bakhitov2022causal},we consider two functions for $g(x)$:
\begin{enumerate}
    \item[a.] log(): $g(x) =  log(|16x - 8| + 1)\text{sign}(x - 0.5)$
    \item[b.] sin(): $g(x) =   sin(x)$
\end{enumerate}
The log transformation is common in empirical studies, to capture a diminishing sensitivity of a feature on the market share. The sine transformation captures periodic or cyclical effects. For example, when the feature represents the time of year, normalized from 0 to 1, then using sine transformation can effectively capture the seasonal variations in consumer preferences.  

The utility that consumer $i$ in market $t$ derives from product $j$ then has the following non-linear form:
\begin{equation}
u_{ijt} = \alpha_i p_{jt} + \beta_{i} g(X_{jt}) + \varepsilon_{ijt}.
\end{equation} 
The marketshares then follow a similar structure to that from Equation \eqref{eq:rcl_mktshare}.

As before, we generate data using this model and estimate the marketshares and elasticities using both our approach and the baseline models. When estimating the baseline MNL and RCL models, we assume that the researcher does not have knowledge of the non-linearities in the utility function and hence uses the simple linear utility(s) in their estimation (as shown in Equation \eqref{eq:linutil_rcl}). The results for the predicted market shares ($\hat{\pi}_{jt}$), own-price elasticity ($\frac{\partial \hat{\pi}_{jt} / \pi_{jt}  }{\partial p_{jt}/p_{jt}}$) and cross-elasticity ($\frac{\partial \hat{\pi}_{jt} / \pi_{jt}  }{\partial p_{k\neq j t}/  p_{k\neq j t} }$) are presented in Tables \ref{tbl:nl_err}, \ref{tbl:nl_elas_s}, and \ref{tbl:nl_elas_c}, respectively. 

Regarding the MAE of predicted market shares, our model surpasses the RCL model by a factor of 8X and 4X across transformations (a) and (b), respectively. Similarly, considering the MAE of predicted own-elasticity in transformations (a) and (b), our model outperforms RCL by factors of 20X and 2.5X, respectively. For the MAE of predicted cross-elasticity, our model is 2X and 1.5X superior to RCL across transformations (a) and (b), respectively. It's worth noting that while our model consistently outperforms the NP method across metrics, the NP method still shows better performance than both RCL and MNL in terms of MAE and RMSE for estimated own-elasticity ($\frac{\partial \hat{\pi}_{jt} / \pi_{jt}  }{\partial p_{jt}/p_{jt}}$), underscoring the strengths of neural network-driven approaches in navigating non-linearities. 

\begin{table}[htp!]
\caption{Predictive Performance in Non-Linear Utility: Market Shares, Own-Elasticity, and Cross-Elasticity}
\label{tbl:non_linearity_combined}

\begin{subtable}{\textwidth}
\centering
\caption{Market Shares ($\hat{\pi}_{jt}$)}
\label{tbl:nl_err}
\scalebox{0.7}{
\begin{tabular}{lccccccccccccccc}
\hline
\hline
    \# & True Model          & \multicolumn{2}{c}{Our model} & \multicolumn{2}{c}{MNL} & \multicolumn{2}{c}{RCL} & \multicolumn{2}{c}{NP}  & \multicolumn{2}{c}{Mean}& No. Obs.\\
    &   &   MAE & RMSE & MAE & RMSE & MAE & RMSE & MAE & RMSE  & MAE & RMSE  &\\
\hline
  0 & RCL-log() &  0.0025 & 0.0063 &   0.0358 & 0.0361 &   0.0213 & 0.0309 & 0.0588 & 0.1235 & 0.0836 &  0.1401 & 200 \\
  1 & RCL-sin() &   0.0029 & 0.0046 &   0.0281 & 0.0340 &   0.0102 & 0.0172 & 0.0315 & 0.0449 & 0.0388  & 0.0527 & 200 \\
\hline
\hline
\end{tabular}}
\end{subtable}
\smallskip

\begin{subtable}{\textwidth}
\centering
\caption{Own-Elasticity ($\frac{\partial \hat{\pi}_{jt} / \pi_{jt}  }{\partial p_{jt}/p_{jt}}$)}
\label{tbl:nl_elas_s}
\scalebox{0.7}{
\begin{tabular}{lccccccccccccc}
\hline\hline
  \# & True Model & \multicolumn{2}{c}{Our Model} & \multicolumn{2}{c}{MNL} & \multicolumn{2}{c}{RCL} & \multicolumn{2}{c}{NP} & No. Obs \\
  & &  MAE & RMSE & MAE & RMSE & MAE & RMSE & MAE & RMSE & \\
\hline
  0 & RCL-log()  & 0.0566 & 0.1523 & 5.4278 & 2.8249 & 1.1961 & 2.1135 & 0.6119 & 0.9085 & 16000 \\
  1 & RCL-sin() &   0.0609 & 0.2820 & 0.6229 & 1.1350 & 0.1777 & 0.4246 & 0.4057 & 1.0787 & 16000 \\
\hline\hline
\end{tabular}
}
\end{subtable}
\smallskip

\begin{subtable}{\textwidth}
\centering
\caption{Cross-Elasticity ($\frac{\partial \hat{\pi}_{jt} / \pi_{jt}  }{\partial p_{k\neq j t}/  p_{k\neq j t} }$)}
\label{tbl:nl_elas_c}
\scalebox{0.7}{
\begin{tabular}{lccccccccccccc}
\hline\hline
  \# & True Model & \multicolumn{2}{c}{Our Model} & \multicolumn{2}{c}{MNL} & \multicolumn{2}{c}{RCL} & \multicolumn{2}{c}{NP}  & No. Obs \\
  & &  MAE & RMSE & MAE & RMSE & MAE & RMSE & MAE & RMSE & \\
\hline
  0 & RCL-log() &   0.0150 & 0.0543 & 0.0436 & 0.1389 & 0.0372 & 0.4414 & 0.2552 & 0.5444 & 144000 \\
  1 & RCL-sin() &   0.0226 & 0.1047 & 0.0471 & 0.1794 & 0.0354 & 0.1751 & 0.1448 & 0.3357 & 144000 \\
\hline\hline
\end{tabular}}
\end{subtable}

\smallskip
\footnotesize
\raggedright
Note: This table presents the results when we add non-linear transformation in data generation. We generate using the Random Coefficient Logit (RCL) model, with 10 products and 100 markets, while only considering a single non-linearly transformed feature, which is the price.  
\end{table}

\subsubsection{Models with Consumer Inattention and Consideration Set Formation}
\label{sssec:consumer_inattention}
Finally, we consider a scenario where consumers do not pay attention to all the products and/or are not fully informed of all the alternatives in the choice set. Recent literature has shown that this is often the case in many empirical settings \citep{goeree2008limited, gabaix2019behavioral, honka2019empirical, abaluck2020method, compiani2022market}. However, such cases violate a standard assumption of the choice model: that consumers are informed and consider all options when they make purchase decisions.  
In some parametric models \citep[e.g.,][]{van2010retrieving}, this issue is managed by constructing a consumer-level consideration set. However, consideration sets are usually unobserved in data; so these approaches often require assumptions on how consideration sets are formed, which might not always be appropriate or reflective of actual consumer behavior. Another way to manage this issue is to model search costs, i.e., allow consumers to ignore certain products because search is costly \citep[e.g.,][]{weitzman1978optimal, mehta2003price, hortaccsu2004product, kim2010online}. Similarly, it also requires researchers to specify how search cost enters the utility function and decision process. In contrast to these models, our approach refrains from making any parametric assumptions, which allows for a potentially more flexible representation of consumer behavior in the case of consumer attention. To demonstrate how our model can capture the inattentive behavior, we look at a scenario where consumers are inattentive and deviate from the traditional random coefficient logit model. 

\begin{figure}[htbp!]
\centering
\begin{subfigure}{\textwidth}
\centering
\includegraphics[width=0.8\textwidth]{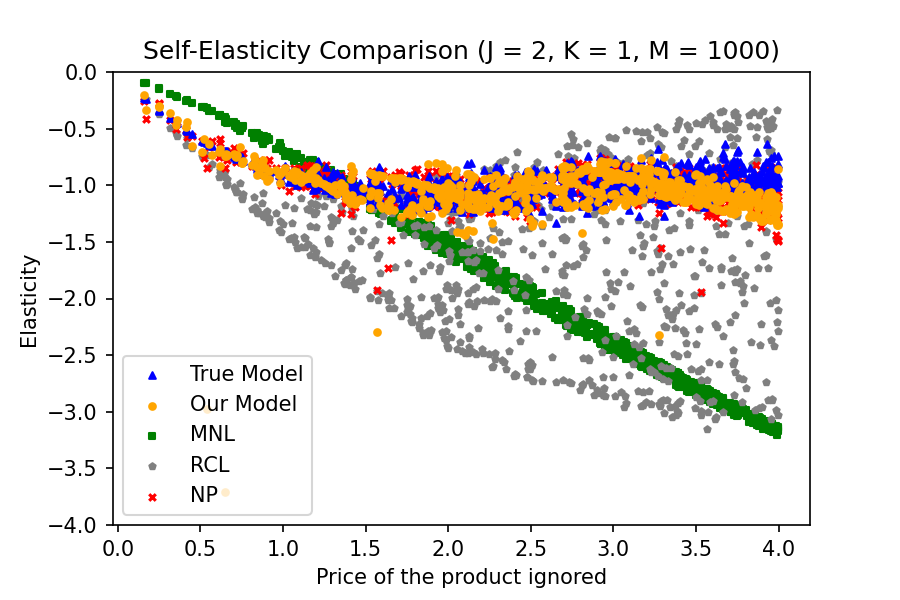}
\caption{Own-elasticity ($\frac{\partial \hat{\pi}_{jt} / \pi_{jt}  }{\partial p_{jt}/p_{jt}}$) in consumer inattention}
\label{fig:in_own_elas}
\end{subfigure}
\begin{subfigure}{\textwidth}
\centering
\includegraphics[width=0.8\textwidth]{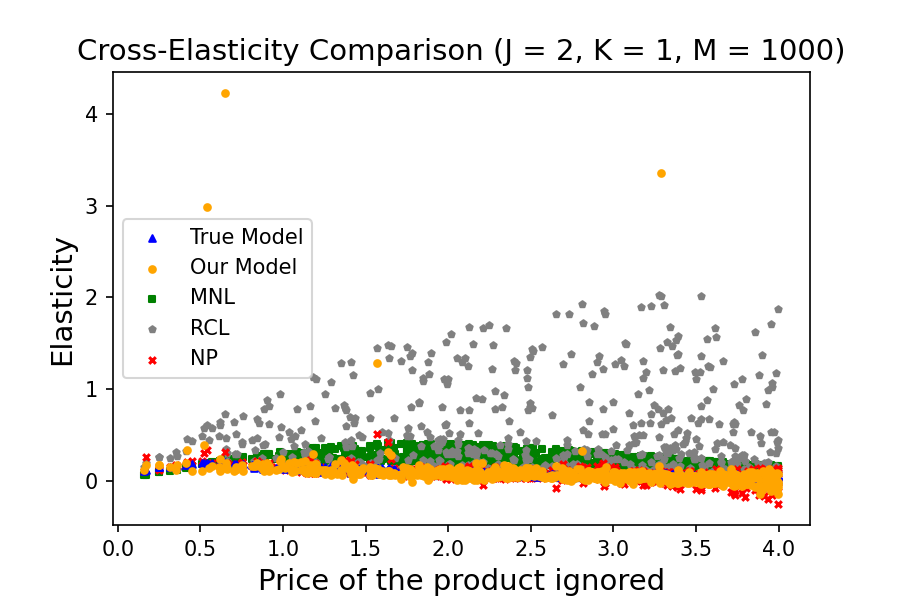}
\caption{Cross-elasticity   ($\frac{\partial \hat{\pi}_{jt} / \pi_{jt}  }{\partial p_{k\neq j t}/  p_{k\neq j t} }$) in consumer inattention}
\label{fig:in_cross_elas}
\end{subfigure}
\caption{Elasticity Effects in Consumer Inattention}
\label{fig:in_elasticity}
\raggedright
\footnotesize{Note: Figure \ref{fig:in_elasticity} illustrates how different models perform when there is $1 - \frac{1}{1+p_{ht}}$ consumers who ignore the product with the highest price. We simulate market shares in the case of 2 products, 1000 markets, and 1 feature (price). Due to the existence of inattention, when the price is higher, the portion of inattentive consumers is higher. Thus when we change the price, the change in market share is smaller than the case without inattention.  }
\end{figure}

We consider a simple model of inattention where a portion of consumers ($1 - \frac{1}{1+p_{ht}}$) ignore the product with the highest price, following \cite{compiani2022market}. In other words, the consideration set of $1 - \frac{1}{1+p_{ht}}$ consumers in market $t$ excludes the highest-price product $h$.  So when the price increases, the portion of inattentive consumers increases. Suppose there is only one feature, price, then, the choice probability of the highest price product $h$ in the market $t$ is:
\[
\frac{1}{1+p_{h t}}\frac{exp(\alpha_i p_{ht} )}{  \sum_k exp(\alpha_i p_{kt} )}.
\]
The choice probability of other products $j \neq h$ are given by:
\[
\frac{1}{1+p_{h t}}\frac{exp(\alpha_i p_{j t} )}{  \sum_k exp(\alpha_i p_{kt}) } + (1- \frac{1}{1+p_{h t}}) \frac{exp(\alpha_i p_{jt})}{ \sum_{k \neq h t } exp(\alpha_i p_{kt} )}. 
\]


\begin{table}[htp!]
\caption{Consumer Inattention - Model Performance of Predicted Market Shares, Own-Elasticity, and  Cross-Elasticity}
\label{tbl:in_summary}

\centering

\subcaption{Market Shares ($\hat{\pi}_{jt}$)}
\label{tbl:in_err}
\scalebox{0.7}{
\begin{tabular}{clccccccccccccccc}
\hline\hline
    \#              & J  & \multicolumn{2}{c}{Our Model} & \multicolumn{2}{c}{MNL} & \multicolumn{2}{c}{RCL} & \multicolumn{2}{c}{NP} & \multicolumn{2}{c}{Mean} & No. Obs \\
\hline
                &   & MAE & RMSE & MAE & RMSE & MAE & RMSE & MAE & RMSE & MAE & RMSE & \\
\hline
0 & 2  & 0.0047 & 0.0198 & 0.0590 & 0.0753 & 0.0316 & 0.0439 & 0.0076 & 0.0273 &    0.1840 & 0.2044 & 40 \\
1 &  5  & 0.0064 & 0.0139 & 0.0203 & 0.0252 & 0.0178 & 0.0226 & 0.0250 & 0.0345 & 0.0780 & 0.1022 & 100 \\
2 &10  & 0.0033 & 0.0068 & 0.0137 & 0.0166 & 0.0072 & 0.0100 & 0.0258 & 0.0365 &  0.0492 & 0.0656 & 200 \\
\hline\hline
\end{tabular}}

\bigskip 

\subcaption{Own-Elasticity ($\frac{\partial \hat{\pi}_{jt} / \pi_{jt}  }{\partial p_{jt}/p_{jt}}$)}
\label{tbl:in_elas_s}
\scalebox{0.7}{
\begin{tabular}{clccccccccccccccc}
\hline\hline
\# &   J &  \multicolumn{2}{c}{Our Model} & \multicolumn{2}{c}{MNL} & \multicolumn{2}{c}{RCL} & \multicolumn{2}{c}{NP} & No. Obs \\
\hline
                           &     &   MAE   &   RMSE   &   MAE   &   RMSE   &   MAE   &   RMSE   &   MAE   &   RMSE   &   \\
\hline
0 &   2  & 0.0609 & 0.5190 & 0.6758 & 1.5698 & 0.3929 & 1.4598 & 0.0573 & 0.8804 & 3200 \\
1 &   5  & 0.0978 & 1.8579 & 0.6917 & 1.9614 & 0.3753 & 1.9782 & 0.4288 & 1.9546 & 8000 \\
2 &  10  & 0.0708 & 2.4273 & 0.8306 & 2.1031 & 0.1842 & 2.2787 & 0.5464 & 2.1450 & 16000 \\
\hline\hline
\end{tabular}}

\bigskip 

\subcaption{Cross-Elasticity ($\frac{\partial \hat{\pi}_{jt} / \pi_{jt}  }{\partial p_{k\neq j t}/  p_{k\neq j t} }$)}
\label{tbl:in_elas_c}
\scalebox{0.7}{
\begin{tabular}{cccccccccccccccc}
\hline\hline
\# &   J &    \multicolumn{2}{c}{Our Model} & \multicolumn{2}{c}{MNL} & \multicolumn{2}{c}{RCL} & \multicolumn{2}{c}{NP} & No. Obs \\
\hline
  &                            &        MAE   &   RMSE   &   MAE   &   RMSE   &   MAE   &   RMSE   &   MAE   &   RMSE   &   \\
\hline
0 &    2 & 0.0419 & 4.3553 & 0.1911 & 8.8197 & 0.3350 & 8.5791 & 0.0378 & 5.6814 & 3200 \\
1 &   5 & 0.0486 & 4.7745 & 0.0827 & 5.8357 & 0.0503 & 5.8404 & 0.1897 & 5.5601 & 32000 \\
2 &  10 & 0.0212 & 3.8765 & 0.0476 & 4.0961 & 0.0146 & 4.1113 & 0.1676 & 4.1198 & 144000 \\
\hline\hline
\end{tabular}}

\bigskip
\footnotesize
\noindent
\raggedright
Note: These tables present the MAE and RMSE for predicted market shares, own-elasticity, and cross-elasticity under consumer inattention across scenarios with 2, 5, and 10 products. The market share and elasticity are simulated assuming a portion of consumers ($1 -\frac{1}{1+p_{j_{ht}}})$, $j_{ht}$ is the index of the highest-price product in market $t$) ignore the product with the highest price. Each scenario is fixed at 100 markets and 1 feature (price) to maintain consistency with the RCL baseline. 

\end{table}
We first present detailed results on the own- and cross-elasticity for the two-product case in Figure \ref{fig:in_elasticity}. We consider the number of markets to be 1,000 so that we can observe enough variance in our data. Figure \ref{fig:in_own_elas} shows how the estimated own-elasticity for the highest-priced product (which is ignored by a subset of consumers), for a range of prices. Note that in our model, when the price is higher, the portion of inattentive consumers is higher. Thus when we change the price, the change in market share is smaller than the case without inattention. While our model and the fully NP model are able to capture this pattern, both the parametric models (MNL and RCL) are unable to do so.  Figure \ref{fig:in_cross_elas} shows how the estimated cross-elasticity of the other products vary with the price of the highest-priced product. Similarly, due to the ignorance of inattention, both MNL and RCL overestimate the magnitude of the elasticity of the other product. In contrast, our model and the fully NP are able to capture the true cross-price elasticity and are close to the true model.

Next, we show a more comprehensive set of results for all three metrics (market-shares, own-, and cross-price elasticities) when there are more products (2, 5, and 10) and fewer markets (100) in Table \ref{tbl:in_err}, \ref{tbl:in_elas_s}, and \ref{tbl:in_elas_c}. We find that our approach consistently outperforms RCL and MNL, as expected. Further, the performance of our model improves as the number of products increases while that of the NP model monotonically decreases with the number of products (for the reasons discussed in $\S$\ref{sssec:mnl_rcl}). 

In summary, we find that our approach adapts well even as the underlying model of consumer behavior changes without the need to impose any specific assumptions on consumer decision-making.

\subsection{Counterfactual Analysis}
\label{ssec:counterfactual}


In general, a key advantage of parametric models like MNL and RCL, or a structural approach to consumer behavior is their ability to predict outcomes in counterfactual scenarios outside the distribution of the data used for estimation. In contrast, fully non-parametric approaches like the naive neural network model that simply uses all product features as inputs cannot be used for counterfactual predictions\footnote{Both our model and fully NP can do counterfactuals when shocks only result in the change in features $\{X_{jt}, P_{jt}\}$. For example, consider a choice model where ranking is a feature. In such a scenario, researchers would like to see how the demand would change if the ranking policy is changed. To estimate this counterfactual scenario, one would update the value assigned to the ranking feature within a model. Therefore, both our model and fully NP are capable of doing such analysis, offering insights into how changes in specific features could impact demand.}. By leveraging the choice invariance property, our method adds some structure to the neural network architecture; and as a result, accommodates certain types of counterfactual models. In particular, we focus on a specific type of counterfactual that is of interest to firms and policy-makers -- demand estimation when the choice set changes; e.g., through the introduction of a new product, the exit of an existing product, or the merger of two firms \citep{nevo2000mergers, petrin2002quantifying, nevo2003new, wollmann2018trucks}. Our model can easily handle such counterfactuals since the choice function is specified as a function of the focal product features and the features of a set of competing products (see Theorem \ref{thm1} and Figure \ref{fig:NN}). In contrast, estimating counterfactual demand when the choice set changes is infeasible with a standard neural network estimator due to its structural constraints on the input space. A change in the choice set would result in the change of size of the input vector, making such estimation infeasible. 


To showcase the capability of our model to estimate counterfactuals, we consider a setting where a new product is introduced to the market. For comparison, we will only consider an MNL and RCL estimator since it is infeasible for the NP method to estimate the counterfactual. We use two data generation processes --  Multinomial Logit (MNL) and Random Coefficient Logit (RCL), the same as $\S$\ref{sssec:mnl_rcl} with 100 markets, 10 products (and one outside option)m and 10 features. For each of these two cases, we consider a counterfactual where an 11th product is introduced to each market. The observable characteristics of the new product are simulated from the same distribution as other products. 

In Table \ref{tbl:counter1}, we present the error in the estimated market share of all products when a new product is introduced to the market. We find that our model does quite well. When the true model is MNL, both MNL and RCL outperform our method, similar to what we find in $\S$\ref{sssec:mnl_rcl}. 
Our model also outperforms the MNL when the underlying data generation process is RCL, and produces results comparable to the true model. Overall, these results suggest that our approach can be used for counterfactual estimation, which is often a key focus of many substantive studies in marketing and economics.

\begin{table}[htp!]
\caption{New Product Demand Estimation - Predicted Market Share ($\hat{\pi}_{jt}$)}
\label{tbl:counter1}
\centering\scalebox{0.7}{
\begin{tabular}{ccccccccc}
\hline\hline
True Model & \multicolumn{2}{c}{Our Model} & \multicolumn{2}{c}{MNL} & \multicolumn{2}{c}{RCL}  & No. Obs.  \\

 &   MAE   &   RMSE   &   MAE   &   RMSE   &   MAE   &   RMSE  & \\
    \hline
MNL  &      0.0234	 & 0.0644   &  0.0041  & 0.0095 &     0.0023 & 0.0045	   &22,000\\
RCL  &      0.0186	 & 0.0145   &  0.0265 & 0.0331  &    0.0023	 & 0.0031      &22,000 \\
\hline\hline
\end{tabular}}
\par
\smallskip 
\footnotesize 
\noindent
\raggedright
Note: This table presents the MAE and RMSE of predicted market shares of all products in the market when a new product enters. We simulate market shares as in our baseline scenario (10 products and one outside option, 100 markets, 10 non-price features) when an eleventh product is introduced.  
\end{table}

\subsection{Inference and Coverage Analysis}
\label{ssec:inf_sim}
We now demonstrate the performance of the debiasing and inference procedure discussed in $\S$\ref{ssec:inference}. The objective is to demonstrate the validity of the estimated confidence intervals. To this end, we estimate the average effect of a 1\% change in own price on demand over all products ($\hat{\theta}$) and compute the corresponding confidence intervals of this effect. The difference between this section and $\S$\ref{sssec:mnl_rcl} lies in both the estimators and the methods. In terms of estimators, in $\S$\ref{sssec:mnl_rcl}, we predict the market share ($\hat{\pi}_{jt}$), own-elasticity ($\frac{\partial \hat{\pi}_{jt} / \pi_{jt}  }{\partial p_{jt}/p_{jt}}$) and cross-elasticity ($\frac{\partial \hat{\pi}_{jt} / \pi_{jt}  }{\partial p_{k\neq j t}/  p_{k\neq j t} }$) for individual products. In contrast, the object of interest in this section is the \textit{average} effect of price on demand across all products ($\hat{\theta}$). As a result, in $\S$\ref{sssec:mnl_rcl}, we did not use the debiasing techniques that we apply here. It is important to emphasize that, in our approach, constructing a confidence interval is viable only for aggregate measures, not for individual observations. 



To simulate the data, we consider a random coefficient logit model of demand with 3 products across 100 markets. We set the true model parameters to be $\beta_{ik} \sim \mathcal{N}(1, 0.5), \alpha_{i} \sim \mathcal{N}(-1,0.5)$. The effect of a 1\% increase in a product's price is given by 
$$\theta_0 = \mathbb{E}[m(w_{jt},\pi)] = \mathbb{E}[\pi(p_{jt}*(1.01), X_{jt}, \{p_{kt}, X_{kt}\}_{k \ne j}) - \pi(p_{jt}, X_{jt}, \{p_{kt}, X_{kt}\}_{k \ne j})],$$ As discussed earlier, one way to estimate this effect is to compute the sample analog of this using the estimated $\hat{\pi}$, such that $\hat{\theta} = \frac{1}{n}\sum_{i = 1}^{n}m(w_{jt},\hat{\pi})$. However, as we pointed out earlier,  the distribution of $\hat{\theta}$ might not be asymptotically normal. To demonstrate this, in Figure \ref{fig:no_correct}, we display the histogram of the estimated effect across 50 random samples by using the plug-in method. We standardize the estimates by subtracting the mean and then dividing by the standard deviation and plot them against the standard normal distribution. As one can observe the distribution appears multi-modal and deviates from a standard normal distribution. Next, we use our proposed debiased estimator and plot the standardized estimates across 50 samples of draws in Figure \ref{fig:with_correct}. The resultant distribution with the debiased estimator is much closer to a standard normal. Finally, we calculate the 95\% confidence intervals using our debiased estimator across 50 random draws. In Table \ref{tab:inference}, we report the bias (mean absolute error from all draws) and the coverage i.e., the percentage of times the true parameter is covered in the estimated confidence intervals. We find that bias across both data-generating processes (RCL and MNL) is notably low, reflecting only a -0.0001 difference from the true effect. The coverage rate of the 95\% confidence interval in our corrected model is 90\%, indicating good coverage. This shows that our debiased estimator can be used to conduct valid inference in finite samples.

\begin{figure}[t]
\centering
\begin{subfigure}{0.49\textwidth}
\includegraphics[width=\textwidth]{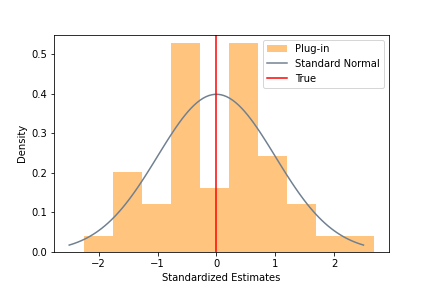}
\caption{Distribution of Estimated Average Effect of Price Change with Plug-in }
\label{fig:no_correct}
\end{subfigure}
\begin{subfigure}{0.49\textwidth}
\includegraphics[width=\textwidth]{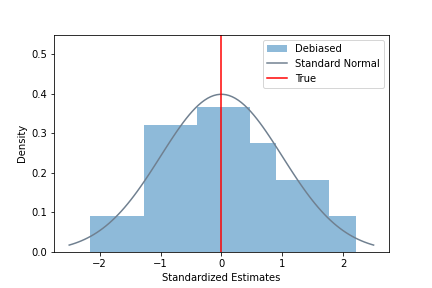}
\caption{Distribution of Estimated Average Effect of Price Change with the Debiased Estimator}
\label{fig:with_correct}
\end{subfigure}
\caption{Distribution of Estimated Average Effect of Price Change}
\label{fig:inference}
\raggedright
\footnotesize{Note: The figure shows the distribution of the standardized plug-in and debiased estimators of the average effect of 1\% change in price on demand. To simulate the data, we consider 3 products across 100 markets with 5 non-price features using  RCL for 50 samples. For each sample, we set the true model parameters to be $\beta_{ik} \sim \mathcal{N}(1, 0.5), \alpha_{i} \sim \mathcal{N}(-1,0.5)$. Figure \ref{fig:no_correct} displays the distribution of the estimated effect with the Plug-in method and Figure \ref{fig:with_correct} shows the result when employing the debiased estimator.  }
\end{figure}

\begin{table}[h!]
    \centering
    \caption{Inference Coverage Analysis}
    \scalebox{0.7}{
    \begin{tabular}{cccc}
    \hline\hline
         True Model & True Effect & Bias   & 95\% CI Cov. \\
    \hline
         RCL &   -0.0013 &   -0.0001  &  90\% \\
         MNL &   -0.0016  &  -0.0001   & 90\% \\
    \hline\hline
    \end{tabular}}
\par
\smallskip 
\footnotesize 
\noindent\raggedright
Note: This table presents the bias and coverage rate of 95\% confidence interval using our debiased estimator of the average effect of 1\% change in price on demand.  To simulate the data, we consider 3 products across 100 markets with 5 non-price features using  RCL and MNL for 50 samples of draws. For each RCL sample, we set the true model parameters to be $\beta_{ik} \sim \mathcal{N}(1, 0.5)$, $\alpha_{i} \sim \mathcal{N}(-1,0.5)$.  For each MNL sample, we set the true model parameters to be $\beta_{ik} = 1, \alpha_{i} =-1$.       
    \label{tab:inference}
\end{table} 

\section{Emprical Data Analysis: US Automobile Data (1971 - 1990)}
\label{sec:empirical_blp}

In this section, we apply our model to a real-world dataset. We use the ``US Automobile Data (1971 -- 1990)"  from \citet{berry1995automobile}. The dataset features cars in the US market from 1971 to 1990, with each year regarded as a  market. The number of cars varies from 86 to 150 each year. For each car, the dataset provides information such as the car's name, the manufacturing company, factory region, market share, price, and four exogenous car characteristics: horsepower, space, mileage per dollar, and the presence of an air conditioning device. 

Even though the dataset is relatively small, it presents three key challenges that make it difficult to use naive non-parametric estimators: (i) the dataset features markets with more than 100 products and only 20 markets in total, (ii) the product assortment in each year or market varies; (iii) the feature ``price" has the endogeneity issue, which was not considered in numerical experiments in $\S$\ref{sec:simulation}. In this section, we demonstrate the use of our estimator, which is capable of effectively addressing such challenges posed in real-world datasets. 

For each component of our model ($\phi_1$, $\phi_2$, and $\rho$), we use a standard 3-layer neural network, and this is implemented without further hyperparameter tuning. We implement ReLU activation function at each layer. This architecture is the same as the one we used in our numerical experiments. For comparison, we replicate the random coefficient logit model (with only the demand side) used by \cite{berry1995automobile} using the Python package \texttt{pyblp} \parencite{PyBLP}. In our replication, we allow for heterogeneity in random coefficients across all variables. Our findings show that the estimates obtained from our model are comparable to the random coefficient logit estimation presented in \cite{berry1995automobile}. 

We estimate our model both without and with consideration of endogeneity. To address endogeneity, we utilize three sets of IVs -- (i) the sum of characteristics of all car models, excluding the product in focus, produced by the same firm in the same year; (ii) the sum of characteristics of all car models, excluding the product in focus, produced by rival firms in the same year; and (iii) cost shifters, which encompass the wage and exchange rate prevalent in the year and region where the factory is located. The utilization of traditional BLP-style instruments, as discussed by \textcite{gandhi2019measuring}, can be problematic due to their relative weakness, often resulting in considerable bias in the estimation of parameters. These issues are significantly exacerbated in non-parametric models. Thus, to counter potential concerns related to weak instruments, we employ a machine-learning-based IV methodology (MLIV) as proposed by \textcite{singh2020machine}. We detail the estimation procedure and results using BLP style IVs in Web Appendix \ref{appendix_mliv}.

In Figure \ref{fig:IV_diff}, we present the estimated own-elasticity ($\frac{\partial \hat{\pi}_{jt} / \pi_{jt}  }{\partial p_{jt}/p_{jt}}$) of our model without IV and with IV. The x-axis represents the price of the focal product, while the y-axis shows the product's own-elasticity. Each point corresponds to a product in a market, resulting in 2,217 observations. We report the estimated elasticity based on the same price variation used in the BLP paper (a 1,000-dollar change). In Figure \ref{fig:noIV}, we observe that the majority of low-priced products (priced below 6,000 dollars) exhibit positive estimated own-elasticity, demonstrating the existence of the endogeneity. We also notice that this issue is attenuated when we correct for endogeneity (see Figure \ref{fig:IV_diff}), i.e., which suggests that our approach is able to handle situations with endogenous features. 

\begin{figure}[t]
\centering
\begin{subfigure}{0.49\textwidth}
\includegraphics[width=\textwidth]{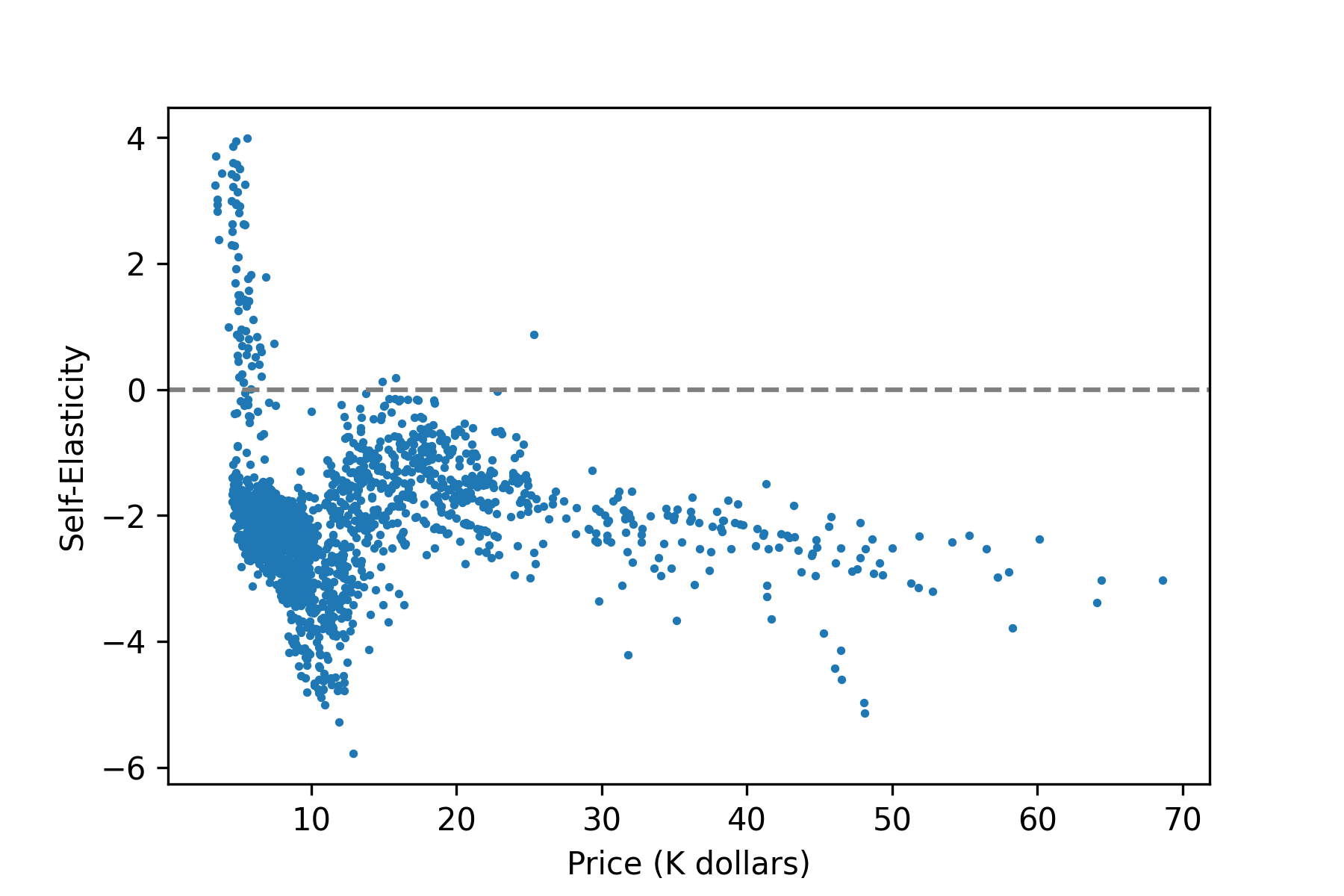}
\caption{Estimated own-Elasticity ($\frac{\partial \hat{\pi}_{jt} / \pi_{jt}  }{\partial p_{jt}/p_{jt}}$)without IV}
\label{fig:noIV}
\end{subfigure}
\begin{subfigure}{0.49\textwidth}
\includegraphics[width=\textwidth]{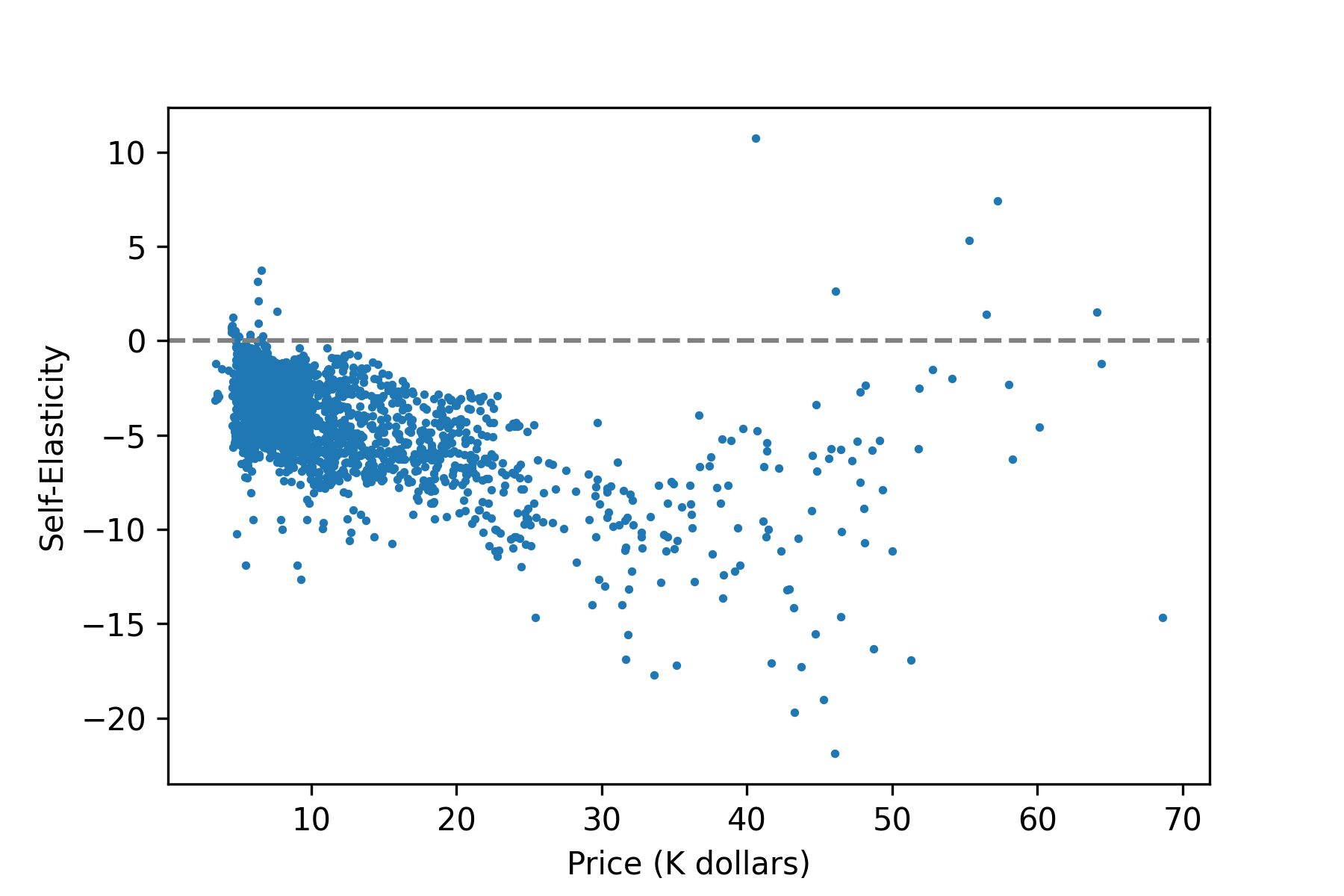}
\caption{Estimated own-Elasticity ($\frac{\partial \hat{\pi}_{jt} / \pi_{jt}  }{\partial p_{jt}/p_{jt}}$) with IV}
\label{fig:withIV}
\end{subfigure}
\caption{Elasticity Estimation Comparison}
\label{fig:IV_diff}
\raggedright
\footnotesize{Figure \ref{fig:IV_diff} presents the estimated own-elasticity ($\frac{\partial \hat{\pi}_{jt} / \pi_{jt}  }{\partial p_{jt}/p_{jt}}$) of our model without IV and with IV. The x-axis represents the price of the focal product, while the y-axis shows the products' own-elasticities. Each point corresponds to a product in a market, resulting in 2,217 observations. We report the estimated elasticity based on the same price variation used in the BLP paper (a 1,000-dollar change). Although we cannot ascertain the true value of own-elasticity, it is widely accepted that own-elasticity should generally be negative for most, if not all, products. In Figure \ref{fig:noIV}, we observe that the majority of low-priced products (priced below 5,000 dollars) exhibit positive estimated own-elasticity, demonstrating the existence of the endogeneity.  
}
\end{figure}

We report the own-elasticity ($\frac{\partial \hat{\pi}_{jt} / \pi_{jt}  }{\partial p_{jt}/p_{jt}}$)  and cross-elasticity ($\frac{\partial \hat{\pi}_{jt} / \pi_{jt}  }{\partial p_{k\neq j t}/  p_{k\neq j t} }$) estimated in our model and random coefficient logit model with a sample of 13 cars in the 1990 market in Tables \ref{tab:blp_our} and \ref{tab:blp_blp}. The sample of 13 cars is the same as the one reported in \cite{berry1995automobile}. Overall, our results are very similar and comparable to \cite{berry1995automobile}. We also plot the distributions of the estimated own-elasticity ($\frac{\partial \hat{\pi}_{jt} / \pi_{jt}  }{\partial p_{jt}/p_{jt}}$)  and cross-elasticity ($\frac{\partial \hat{\pi}_{jt} / \pi_{jt} }{\partial p_{k\neq j t}/  p_{k\neq j t} }$) obtained from our model and the BLP model in Figure \ref{fig:blp_elasticity}. The filled areas in the violin plots represent the complete range of the elasticities, while the text labels next to the line indicate the mean values. The estimated mean own- and cross-elasticities appear to be similar between our model and the BLP model, though our model exhibits a larger RMSE in the estimated elasticity values compared to the BLP model.

\begin{sidewaystable*}
    \centering
\scalebox{0.7}{
    \begin{tabular}{lccccccccccccc}
\hline
{} &  Acura  &  BMW  &  Buick  &  Cadillac  &  Chevy  &  Ford &  Ford  &  Honda  &  Lexus  &  Lincoln  &  Mazda  &  Nissan  &  Nissan \\
{} &  Legend &   735i &  Century &  Seville &  Cavalier & Escort &  Taurus &  Accord &  LS400 &  Town Car & 323 &   Maxima &  Sentra \\
\hline
Acura Legend     &       -5.6060 &    0.1993 &         0.2198 &            0.2221 &          0.0632 &       0.2317 &       0.2199 &        0.2337 &       0.2144 &            0.2187 &     0.2595 &         0.2354 &         0.2595 \\
BMW 735i         &        0.4095 &   -6.1528 &         0.3653 &            0.4093 &          0.0352 &       0.3525 &       0.3655 &        0.3807 &       0.4120 &            0.4084 &     0.3547 &         0.4161 &         0.3547 \\
Buick Century    &        0.1234 &    0.1020 &        -6.1023 &            0.1213 &          0.0376 &       0.1294 &       0.1215 &        0.1205 &       0.1175 &            0.1191 &     0.1400 &         0.1299 &         0.1400 \\
Cadillac Seville &        0.2895 &    0.2179 &         0.2631 &           -7.2896 &          0.0647 &       0.2692 &       0.2631 &        0.2566 &       0.2810 &            0.2818 &     0.2892 &         0.2849 &         0.2892 \\
Chevy Cavalier   &        0.0142 &   -0.0013 &         0.0202 &            0.0167 &         -1.3447 &       0.0171 &       0.0202 &        0.0353 &       0.0126 &            0.0167 &     0.0354 &         0.0291 &         0.0355 \\
Ford Escort      &        0.0410 &    0.0245 &         0.0353 &            0.0413 &         -0.0148 &      -1.8519 &       0.0353 &        0.0520 &       0.0392 &            0.0413 &     0.0494 &         0.0513 &         0.0494 \\
Ford Taurus      &        0.1166 &    0.0914 &         0.1188 &            0.1160 &          0.0183 &       0.1199 &      -6.1473 &        0.1258 &       0.1066 &            0.1157 &     0.1451 &         0.1290 &         0.1451 \\
Honda Accord     &        0.0975 &    0.0647 &         0.0968 &            0.1006 &         -0.0003 &       0.0945 &       0.0969 &       -5.7438 &       0.0914 &            0.1006 &     0.1166 &         0.1140 &         0.1165 \\
Lexus LS400      &        0.3357 &    0.2606 &         0.3136 &            0.3325 &          0.0822 &       0.3235 &       0.3137 &        0.3126 &      -6.8791 &            0.3271 &     0.3495 &         0.3348 &         0.3494 \\
Lincoln Town Car &        0.2681 &    0.2310 &         0.2548 &            0.2656 &          0.0713 &       0.2663 &       0.2548 &        0.2648 &       0.2586 &           -5.3996 &     0.3009 &         0.2719 &         0.3009 \\
Mazda 323        &        0.0361 &    0.0212 &         0.0272 &            0.0326 &         -0.0127 &       0.0249 &       0.0272 &        0.0363 &       0.0323 &            0.0326 &    -2.6589 &         0.0404 &         0.0357 \\
Nissan Maxima    &        0.1579 &    0.1367 &         0.1589 &            0.1555 &          0.0425 &       0.1670 &       0.1589 &        0.1689 &       0.1484 &            0.1534 &     0.1884 &        -7.2216 &         0.1884 \\
Nissan Sentra    &        0.0386 &    0.0239 &         0.0304 &            0.0384 &         -0.0202 &       0.0294 &       0.0304 &        0.0496 &       0.0375 &            0.0383 &     0.0439 &         0.0470 &        -1.8754 \\
\hline
\end{tabular}
}
    \caption{Estimated own- and cross-elasticities of a sample of automobile data using our method}
    \label{tab:blp_our}
\par
\smallskip 
\footnotesize 
\noindent
\raggedright

Note: This table presents the estimated own- and cross-elasticity of a sample of 13 cars in the 1990 market using our model. The selected cars are the same as \cite{berry1995automobile} reports. Each entry with row index $i$ and column index $j$ gives the percentage change in demand divided by the percentage change in price (based on  \$1,000 change in the price of $i$).  

\end{sidewaystable*}

\begin{sidewaystable*}
    \centering
\scalebox{0.7}{
    \begin{tabular}{lccccccccccccc}
\hline
{} &  Acura  &  BMW  &  Buick  &  Cadillac  &  Chevy  &  Ford &  Ford  &  Honda  &  Lexus  &  Lincoln  &  Mazda  &  Nissan  &  Nissan \\
{} &  Legend &   735i &  Century &  Seville &  Cavalier & Escort &  Taurus &  Accord &  LS400 &  Town Car & 323 &   Maxima &  Sentra \\
\hline
Acura Legend     &       -5.4677 &    0.0489 &         0.0205 &            0.1029 &          0.0221 &       0.0220 &       0.0143 &        0.1477 &       0.1503 &            0.0273 &     0.0013 &         0.1359 &         0.0039 \\
BMW 735i         &        0.1267 &   -9.8502 &         0.0156 &            0.1058 &          0.0122 &       0.0121 &       0.0057 &        0.1313 &       0.1546 &            0.0278 &     0.0006 &         0.1375 &         0.0022 \\
Buick Century    &        0.0184 &    0.0054 &        -5.1978 &            0.0124 &          0.1153 &       0.1043 &       0.1475 &        0.1982 &       0.0165 &            0.0800 &     0.0076 &         0.0379 &         0.0174 \\
Cadillac Seville &        0.1293 &    0.0513 &         0.0175 &           -6.6819 &          0.0151 &       0.0150 &       0.0099 &        0.1393 &       0.1576 &            0.0271 &     0.0008 &         0.1406 &         0.0027 \\
Chevy Cavalier   &        0.0131 &    0.0028 &         0.0766 &            0.0071 &         -3.1163 &       0.1421 &       0.0849 &        0.2608 &       0.0086 &            0.0404 &     0.0100 &         0.0395 &         0.0241 \\
Ford Escort      &        0.0137 &    0.0029 &         0.0726 &            0.0074 &          0.1487 &      -3.0590 &       0.0603 &        0.2781 &       0.0090 &            0.0263 &     0.0106 &         0.0419 &         0.0258 \\
Ford Taurus      &        0.0048 &    0.0007 &         0.0554 &            0.0027 &          0.0479 &       0.0326 &      -4.0258 &        0.0727 &       0.0017 &            0.1779 &     0.0026 &         0.0122 &         0.0057 \\
Honda Accord     &        0.0387 &    0.0133 &         0.0582 &            0.0291 &          0.1151 &       0.1173 &       0.0568 &       -4.3399 &       0.0409 &            0.0297 &     0.0081 &         0.0618 &         0.0200 \\
Lexus LS400      &        0.1350 &    0.0536 &         0.0166 &            0.1126 &          0.0130 &       0.0130 &       0.0045 &        0.1400 &      -7.4316 &            0.0243 &     0.0006 &         0.1464 &         0.0024 \\
Lincoln Town Car &        0.0087 &    0.0034 &         0.0286 &            0.0069 &          0.0217 &       0.0135 &       0.1693 &        0.0362 &       0.0086 &           -5.6139 &     0.0011 &         0.0123 &         0.0024 \\
Mazda 323        &        0.0114 &    0.0020 &         0.0743 &            0.0056 &          0.1476 &       0.1494 &       0.0679 &        0.2723 &       0.0063 &            0.0304 &    -2.8631 &         0.0390 &         0.0254 \\
Nissan Maxima    &        0.1008 &    0.0393 &         0.0314 &            0.0831 &          0.0493 &       0.0500 &       0.0271 &        0.1749 &       0.1209 &            0.0286 &     0.0033 &        -4.7872 &         0.0086 \\
Nissan Sentra    &        0.0140 &    0.0031 &         0.0707 &            0.0078 &          0.1471 &       0.1504 &       0.0617 &        0.2763 &       0.0095 &            0.0271 &     0.0105 &         0.0422 &        -3.1799 \\
\hline
\end{tabular}
}
    \caption{Estimated own and cross-elasticities of a sample of automobile data using the BLP model}
    \label{tab:blp_blp}

\par
\smallskip 
\footnotesize 
\noindent
\raggedright
Note: This table presents the estimated own- and cross-elasticity of a sample of 13 cars in the 1990 market using the BLP model. The selected cars are the same as \cite{berry1995automobile} reports. Each entry with row index $i$ and column index $j$ gives the percentage change in demand divided by the percentage change in price (based on  \$1,000 change in the price of $i$).  
\end{sidewaystable*}

\begin{figure}[h]
\centering
\begin{subfigure}{0.49\textwidth}
\includegraphics[width=\textwidth]{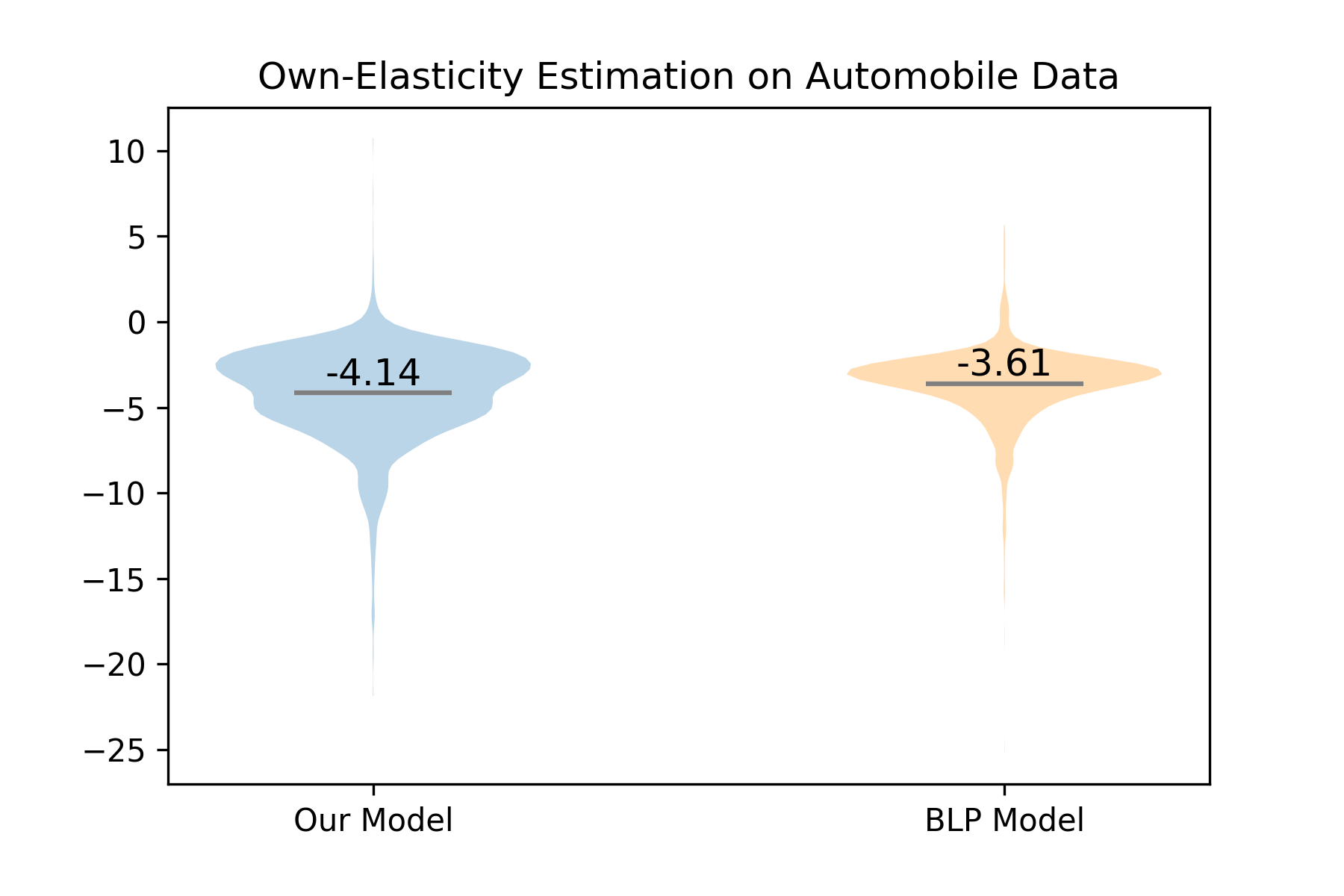}
\caption{Own-Elasticity Estimation (Our Model vs. BLP Model)}
\label{fig:blp_own_elas}
\end{subfigure}
\begin{subfigure}{0.49\textwidth}
\includegraphics[width=\textwidth]{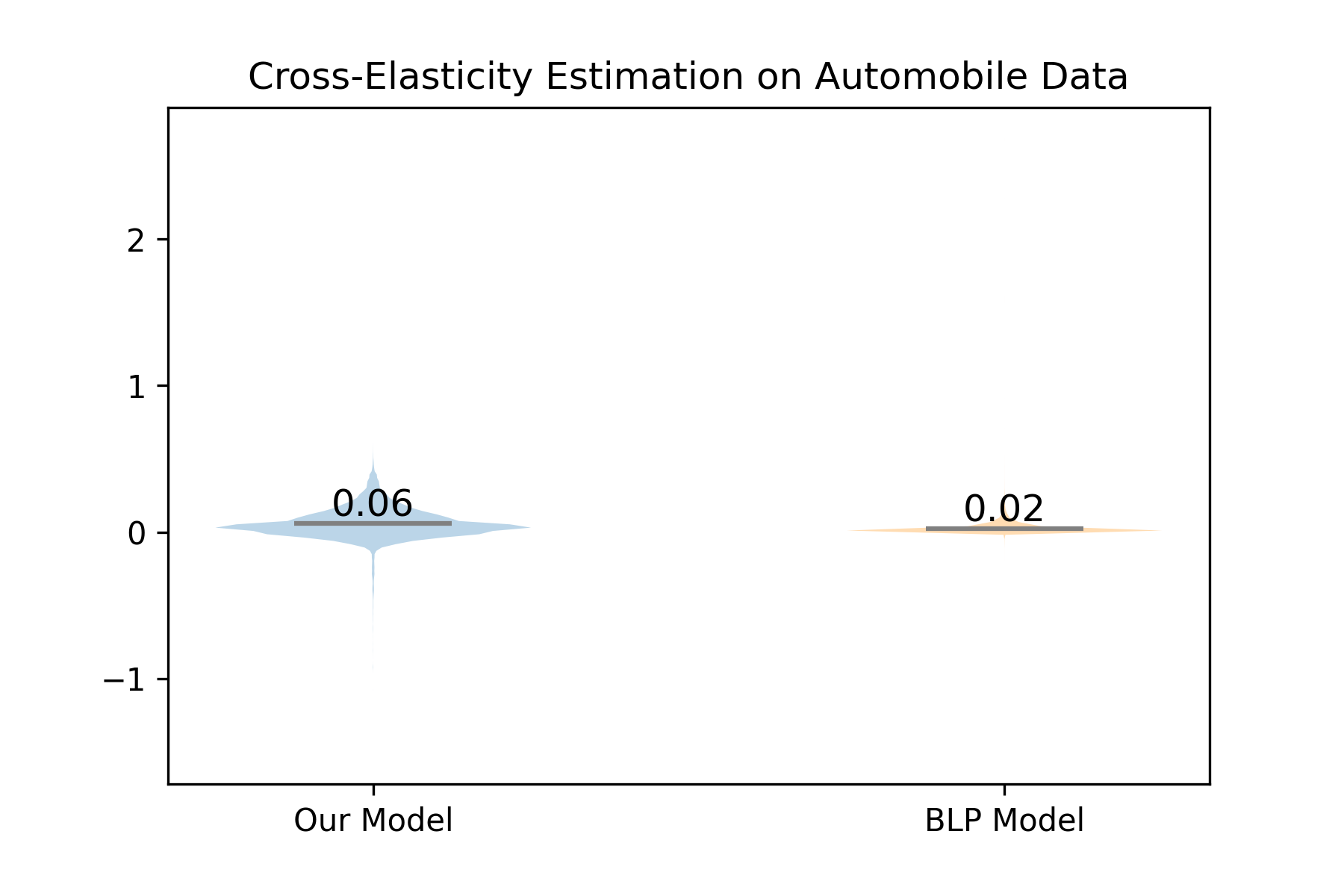}
\caption{Cross-Elasticity Estimation (Our Model vs. BLP Model)}
\label{fig:blp_cross_elas}
\end{subfigure}
\caption{Elasticity Estimation Comparison}
\label{fig:blp_elasticity}

\raggedright
\footnotesize{Note: Figure \ref{fig:blp_elasticity} illustrates the distributions of the estimated own- and cross-elasticities obtained from our model and the BLP model. The filled areas in the violin plots represent the complete range of the elasticities, while the text labels indicate the mean values. }
\end{figure}

We further estimate the average own-elasticity ($\hat{\theta}$) for high-priced, medium-priced, and low-priced cars and construct a confidence interval for each category using our inference procedure. We present our result in Table \ref{tab:inf_blp}. Both our model and the BLP model indicate that the average own-elasticity is highest (in terms of the absolute value of own-elasticity) for high-priced cars and lowest for low-priced cars.  Moreover, the 95\% confidence intervals for all three categories do not include zero. This also demonstrates the efficiency of our model even when there is a limited sample of only 20 observations.


\begin{table}[h]
    \centering\scalebox{0.7}{
    \begin{tabular}{lccr}
    \hline\hline
     & BLP Model & Our Model & No. Obs.  \\
     & Mean Estimate &  Mean Estimate (95\% Confidence Interval) &\\
     \hline
     High &  -5.6705  & -4.1922 (-6.2204, -2.1641) & 20\\
     Medium    & -3.7354 & -3.7215 (-4.9260, -2.5169) & 20   \\
     Low &  -2.9174  & -2.0697 (-3.1980, -0.9415) & 20\\
     \hline\hline
    \end{tabular}}
    \caption{Estimates of Average Own-Elasticity}
    \label{tab:inf_blp}
\par
\smallskip 
\footnotesize 
\noindent
\raggedright
Note: This table presents the estimated average own-elasticity for cars across various price categories. We categorize cars with a price over \$20k as ``high-priced", cars priced between \$8k and \$20k as ``medium-priced", and all other cars as ``low-priced". For each category, we randomly select one car of the category from each market as one observation. For the BLP model, we calculate the average own-elasticity of the sampled cars as the mean estimate. For our model, we estimate the average own-elasticity and construct the confidence interval following our inference procedure. 
\end{table}

The empirical analysis demonstrates the applicability and effectiveness of our model in a real-world setting, addressing challenges such as limited sample size, variability in product assortments, and endogeneity. The comparable results with established econometric models, such as BLP model, help in validating the robustness and reliability of our approach.

\section{Conclusion} 
Choice models are fundamental in understanding consumer behavior and informing business decisions. Over the years, various methods, both parametric and non-parametric, have been developed to represent consumer behavior.  While parametric methods, such as logit or probit-based models, are favored for their simplicity and interpretability, their restrictive assumptions can limit their ability to fully capture consumer preferences' intricacies. On the other hand, non-parametric methods offer a more flexible approach, but they often suffer from the ``curse of dimensionality", where the complexity of estimating choice functions escalates exponentially with an increase in the number of products.

In this paper, we propose a fundamental characterization of choice models that combines the tractability of traditional choice models and the flexibility of non-parametric estimators. This characterization specifically tackles the challenge of high dimensionality in choice systems and facilitates flexible estimation of choice functions. Through extensive simulations, we validate the efficacy of our model, demonstrating its superior ability to capture a range of consumer behaviors that traditional choice models fail to capture.  We also show how to address the endogeneity issue and estimate counterfactuals in our characterization. 
Furthermore,  leveraging the recent strides in the automatic debiased machine learning literature, we offer an inference procedure that constructs confidence intervals on relevant objects, such as price elasticities. Finally, we apply our method to the automobile dataset from \cite{berry1995automobile}. Our empirical analysis affirms that our model produces results that align well with the extant literature.

Our paper opens many avenues for future research. We focus on using neural network-based estimators. However, estimators, such as Gaussian processes and Gradient boosting-based estimators can be adopted to estimate the proposed functionals. Also, we believe more experimentation needs to be done on the neural network design side.

\section*{Competing Interests Declaration}
Author(s) have no competing interests to declare.


\setlength{\bibsep}{0.3em}

\bibliography{ref}
\newpage


\newpage

\begin{appendices}

\setcounter{table}{0}
\setcounter{figure}{0}
\setcounter{equation}{0}
\setcounter{page}{0}
\renewcommand{\thetable}{A\arabic{table}}
\renewcommand{\thefigure}{A\arabic{figure}}
\renewcommand{\theequation}{A\arabic{equation}}
\renewcommand{\thepage}{\roman{page}}
\pagenumbering{roman}

\singlespacing



\section{Identity Independence and Permutation Invariance of Traditional Choice Models}
\label{sec:appendix_per_inv}

In this Web Appendix, we show why choice models listed in Table \ref{tbl:choice_models} satisfy the identity independence and the permutation invariance assumption. 
Identity independence is achieved when all the information about a product utilized in the model is included in its features. To show that these models satisfy the permutation invariance assumption, we first permutate any two competitor products associated with a focal product. We then illustrate that such permutations do not alter the demand for the focal product, thereby establishing the models' compliance with the permutation invariance assumption. 

Since all models require parametric assumptions about the utility function, for the sake of simplicity in notation, we assume that the utility of a product is a parametric function of its observed features $X_{jt}$ and price $p_{jt}$, determined by a set of parameters $\beta $ unless otherwise specified. Mathematically, we express the utility as:
\[
v_{jt} = f(X_{jt}, p_{jt};\beta)
\]

\subsection{MNL}

Assuming there are $J$ products, the market share of each product under the MNL model is

\begin{equation}
    \pi_{jt} = g(X_{jt}, p_{jt} , \{X_{kt}, p_{kt}\}_{k \in S, k \neq j}) = \frac{exp( f(X_{jt}, p_{jt};\beta))}{exp( f(X_{jt}, p_{jt};\beta) + \sum_{k \in \mathcal{S}, k \neq j } exp(f(X_{kt}, p_{kt};\beta))}.
\end{equation}
If we permute two products $k_1$ and $k_2$ in the set $\mathcal{S} \backslash \{j\}$, the only change is in the order of terms in the summation in the denominator of the market share.  Since addition is commutative, the sum $\sum_{k \in \mathcal{S}, k \neq j } exp(f(X_{kt}, p_{kt};\beta))$ remains the same regardless of the order of $k_1$ and $k_2$. Therefore, the value of $\pi_{jt}$ remains unchanged, demonstrating that the MNL model satisfies permutation invariance.

\subsection{Mixed Logit}

Assuming there are $J$ products, the market share of each product under the RCL model is 
\begin{equation}
    \pi_{jt} = g(X_{jt}, \{X_{kt}\}_{k \in S, k \neq j}) = \int \frac{exp( f(X_{jt}, p_{jt};\beta_i))}{exp( f(X_{jt}, p_{jt};\beta_i) + \sum_{k \in \mathcal{S}, k \neq j } exp(f(X_{kt}, p_{kt};\beta_i))}  d F(\beta_i)
\end{equation}
, where  $F(\beta_i)$ is the CDF of random coefficient $\beta_i$. Similar to MNL, if we permute the position of any other two products $k_1, k_2  \in \mathcal{S} \backslash \{j\} $, it only affects the sequence of the summation in the denominator,  thus the demand of product $j$ remains the same. Therefore, it satisfies the permutation invariance. 

\subsection{Nested Logit}

In a nested logit model, products are grouped into nests to account for the correlation in unobserved factors among products within the same nest. Consider a scenario where there are $J$ products categorized into $L$ nests. The market share of a product $j$ in a nest $l$ is calculated taking into account this nesting structure. A key aspect of showing permutation invariance in the Nested Logit model is to include the nest affiliation of each product as a feature. This is represented by an $L$-dimensional 0-1 vector $N_{jt} \in \mathbb{R}^ {1 \times L}$, which denotes the affiliation of product $jt$ to one of the $L$ nests. Each element of the vector corresponds to a specific nest, indicating the product's membership in that nest. Specifically,
\begin{equation}
    \pi_{jt} = g(X_{jt}, p_{jt}, N_{jt}, \{X_{kt}, p_{kt}, N_{kt}\}_{k \in S, k \neq j}).
\end{equation}
Then, the Nested Logit model can be represented as:
\begin{equation}
    \pi_{jt} = P(jt|l) P(l)  
\end{equation}
\begin{equation}
    P(j|l) =  \frac{\exp\left(\frac{ v_{jt}}{N_{jt} \sigma }\right)}{\sum_{k \in \mathcal{N}_l}\exp\left(\frac{ v_{kt}}{N_{jt} \sigma }\right)} 
\end{equation}
\begin{equation}
     P(l) = \frac{\exp\left(N_l \sigma \log\left(\sum_{k \in \mathcal{N}_l}\exp\left(\frac{v_{kt}}{N_{kt} \sigma}\right)\right)\right)}{\sum_{m=1}^{L}\exp\left(N_{m} \sigma  \log\left(\sum_{k \in \mathcal{N}_m}\exp\left(\frac{v_{kt}}{N_{m} \sigma}\right)\right)\right)}
\end{equation}
Where:
\begin{itemize}
    \item $ \mathcal{N}_l $ is the set of products in nest $ l $.
    \item  $\sigma \in R^{L \time 1} $ is an L-dimentional vector where the $l$th  element represents the scale parameter for nest $l$ which captures the correlation in unobserved utilities among products in the same nest.  Therefore, the scale parameter of product $jt$, which belongs to nest $l$, can be expressed as $N_{jt} \sigma$. 
    \item $N_m$ denotes a binary vector of dimension $L$, in which the $m$th element is one and all other elements are zero. Therefore, $N_m \sigma$ is the scale parameter of the nest $m$.
\end{itemize}

It is intuitive that when we swap the position of any two products, the probability of choosing a nest ($P(l)$) would not change because both the denominator and the numerator only depend on the nest affiliations. Moreover, $P(j|l)$ also satisfies permutation invariance because of the additive nature of the denominator in its formula, where the sum over all products in nest $l$ remains constant despite any permutation of product positions within the nest. The permutation invariance property could be easily extended to that random coefficient nested logit following the same logic as random coefficient logit.

\subsection{Latent Class Logit Model}
In a latent class logit model, products can be divided into multiple latent classes and each latent class has a distinct set of parameters. Assume there are $c$ latent classes, the utility of a product in class $c$ is denoted by $v_{jt,c} = f(X_{jt}, p_{jt}; \beta_c)$. The probability of one product belonging to a class $c$ is determined by the linear combination of $Z_{jt}$ with parameters $\alpha_c$. 

\begin{equation}
\begin{aligned}
    \pi_{jt} &= g(X_{jt}, p_{jt}, Z_{jt}, \{X_{kt}, p_{kt}, Z_{kt}\}_{k \in S, k \neq j}) \\
    &= \sum_c  \frac{\exp(f(X_{jt}, p_{jt}; \beta_c))}{\exp(f(X_{jt}, p_{jt}; \beta_c) )+ \sum_{k \in S, k \neq j} \exp(f(X_{kt}, p_{kt}; \beta_c)}  \frac{\exp(\alpha_c Z_{jt})}{\sum_c \exp(\alpha_c Z_{jt})} 
\end{aligned}
\end{equation}

Similarly, the $\sum_{k \in S, k \neq j} \exp(f(X_{kt}, p_{kt}; \beta_c)$ is commutative thus the demand remains the same when we alter the order of any two products except $jt$. Therefore, latent class models satisfy the permutation invariance.

\subsection{Consumer Inattention and Search Model} 

In this part, we specifically demonstrate the permutation invariance property of the consumer inattention model that we used in the numerical experiments. Our analysis could easily be extended to other models based on consumer inattention in the literature, such as \cite{goeree2008limited}, \cite{turlo2023discrete}, \cite{compiani2022market} and \cite{joo2023rational}, as well as search models, such as \cite{mehta2003price}. 

In our simulation, we generate the market shares of products by considering a segment of consumers who ignore the highest-priced product. Here we generalize the specification by assuming the portion of the inattentive consumers is a function of observed features of the highest-priced product $1 - s(X_{ht})$, where $h$ denotes the highest-priced product. A consumer forms their consideration set with the products they pay attention to. For products in their consideration set, they use the logit function as the decision rule. We allow for each consumer to have a set of random coefficients. Then the choice probability of the highest-priced product becomes 
\begin{equation}
    \pi_{ht} =  g(X_{ht}, \{X_{kt}\}_{k\in \mathcal{S}, k\neq h}) =  \frac{s(X_{ht})exp(\beta_i X_{ht})}{f(X_{jt}, p_{jt};\beta_i) + \sum_{k \in \mathcal{S}, k \neq h} f(X_{kt}, p_{kt};\beta_i)}
\end{equation}
When we permutate the order of any two competitor products, the choice probability of the highest-price product remains the same. The choice probability of other products can be expressed as 
\begin{equation}
\begin{aligned}
    \pi_{jt} &= g(X_{jt}, \{X_{kt}\}_{k\in \mathcal{S}, k\neq j}) \\
    &= \frac{s(X_{ht}) exp(f(X_{kt}, p_{jt};\beta_i))}{exp(f(X_{kt}, p_{jt};\beta_i)) + \sum_{k \in \mathcal{S}, k \neq j} f(X_{kt}, p_{kt};\beta_i)} + \frac{(1-s(X_{ht}))  exp(f(X_{kt}, p_{jt};\beta_i))}{exp(f(X_{kt}, p_{jt};\beta_i)) + \sum_{k \in \mathcal{S}, k \neq j, h} f(X_{kt}, p_{kt};\beta_i)} 
\end{aligned}
\end{equation}
Similar to the highest-priced product, the choice probability for other products remains unaffected when the order of competitor products is altered. This is because the feature values of the most expensive product $X_{h}$ do not change even though the index $h$ may change. Therefore, consumer inattention models are also permutation invariant.

\subsection{Other Models}

We have illustrated that the Multinomial Logit model, Nested Logit model, Mixed Logit model, and Latent Class Logit model all exhibit the property of permutation invariance. It becomes apparent that other models, which employ a similar random utility framework but differentiate in the probability distribution of the error term, likewise possess this permutation invariance property.

For example, this permutation invariance property extends Generalized Extreme Value (GEV) models \citep{train2009discrete}, where the error term follows the generalized extreme value distribution, and the Probit model \citep{hausman1978conditional, chintagunta2001endogeneity}, where the error term follows a normal distribution. Although there is no straightforward closed-formed representation of choice probability for these models, the choice probability can be expressed as
\begin{equation}\label{eq:rum}
    \pi_{ijt} = \operatorname{Pr}(u_{ijt} \geq u_{ikt}, \forall k \neq j, k \in \mathcal{S}_t),
\end{equation}
where $u_{ijt}$ represents the consumer $i$'s utility of product $j$ in market $t$. Since $u_{ikt}$ for any $k \in \mathcal{S}_t$ is parameterized by its own feature ($X_{kt}$), permuting any two $k \neq j, k \in \mathcal{S}_t$, does not affect the choice probability of $\pi_{ijt}$ as well as the aggregate demand $\pi_{jt}$. Indeed, this shows that any models building on the random utility framework with the decision rule as stated in Equation \ref{eq:rum} are permutation invariant.

The permutation invariance property also extends to other recently developed choice models, such as Markov Chain Choice model\citep{blanchet2016markov}. The choice probability of a product $jt$ is modeled as a Markov Chain, which consists of two parts: arrival probability $\lambda_{jt}$ and transition probability $p_{jk,t}$, defined as
\begin{align}
    \lambda_{jt} &= f(j, \mathcal{S}_t), \\
    p_{jk,t} &= \begin{cases}
    1, & \text{if } j = 0 \text{ and } k = 0, \\
    \frac{f(j, \mathcal{S}_t \setminus \{j\}) - f(k, \mathcal{S}_t)}{f(j, \mathcal{S}_t)}, & \text{if } j, k \in \mathcal{S}_t, k \neq j, \\
    0, & \text{otherwise.}
    \end{cases}
\end{align}

Given that both $\lambda_{jt}$ and $p_{jk,t}$ only rely on $jt$, $\mathcal{S}_t$ and $\mathcal{S}_t \setminus \{j\}$, permuting the order of any two products that are not the focal product $jt$ in the market $t$ does not change both arrival probability and transition probability. Therefore, the choice probability does not change.

\section{Proof of Main Results} \label{sec:appendix_proof_thrm1}



\mainthm*

Proof. The sufficiency follows by observing that the function $\pi_{jt} = \rho (\phi_1 (X_{jt}, p_{jt}) + \sum_{k \in S \setminus \{j\}} \phi_2 (X_{kt}, p_{kt}))$ satisfies assumption  2 and 3. To prove necessity, first consider 
\(\mathbb{E} = \{2n \mid n \in \mathbb{N}\}\)
 and \(\mathbb{O} = \{2n + 1 \mid n \in \mathbb{N}\}\) as the set of all even and odd natural numbers, respectively. Next, to show that all choice functions satisfying assumptions \ref{assumpt_id} (identity independence) and \ref{assumpt_perminv} (permutation invariance) can be decomposed in the above manner, we begin by noting that there must be a mapping from the elements to the set of even number and odd numbers respectively, since the elements belong to a countable universe $\mathbb{C}^{k}$. Let these mappings be denoted by $c^{e}: \mathbb{C}^{k} \rightarrow \mathbb{E}$ and $c^{o}: \mathbb{C}^{k} \rightarrow \mathbb{O}$. Now if we let $\phi_1(X_{jt})=4^{-c^{e}(X_{jt}, p_{jt})}$ and $\phi_2(X_{kt}, p_{kt}) = 4^{-c^{o}(X_{kt}, p_{kt})}$ then $\phi_1(X_{jt}, p_{jt}) + \sum_{k \in \mathcal{S}_t \setminus \{j\}} \phi_2(X_{kt}, p_{kt})$ constitutes an unique representation for every product $j$ and competing assortment $\mathcal{S}_t \setminus \{j\}$. Now a function $\rho: \mathbb{R} \rightarrow \mathbb{R}$ can always be constructed such that $\pi_{jt}=\rho\left(\phi_1 (X_{jt}, p_{jt}) + \sum_{k \in S \setminus \{j\}} \phi_2 (X_{kt}, p_{kt})\right)$.


\mainprop*

\begin{proof}
    To show the asymptotic normality we will first verify the Assumptions 1-3 of \cite{chernozhukov2022locally}, from now on CEINR, with $g(w, \pi(z;\gamma), \theta) = m(w, \pi(z;\gamma)) - \theta$ and $\phi(w,\pi(z;\gamma),\alpha(z), \theta) = \alpha(z)\cdot(y - \pi(z;\gamma))$. Using Taylor series expansion, Assumption \ref{assmp-2} and $||\hat{\pi}(z;\gamma)-\pi_0(z;\gamma)|| \xrightarrow[]{p} 0$ we have,

\begin{align}
& \int ||g(w, \hat{\pi}(z_i;\hat{\gamma}), \theta_0) - g(w, \pi_0(z_i;\gamma_0), \theta_0)||^2\mathcal{P}_0(dw) \nonumber\\
& = \int ||m(w, \hat{\pi}(z_i;\hat{\gamma})) - m(w, \pi_0(z_i;\gamma_0))||^2\mathcal{P}_0(dw) \nonumber\\
& \le C\int ||\hat{\pi}(z_i;\hat{\gamma})-\pi_0(z_i;\gamma_0)||^2\mathcal{P}_0(dw) \nonumber\\
& \le C\int ||\hat{\pi}(z_i;\hat{\gamma})- \hat{\pi}(z_i;\gamma_0) + \hat{\pi}(z_i;\gamma_0) -
\pi_0(z_i;\gamma_0)||^2\mathcal{P}_0(dw) \nonumber\\
\text{{By the triangle inequality}} \nonumber\\
& \le C\int||\hat{\pi}(z_i;\hat{\gamma})-\hat{\pi}(z_i;\gamma_0)||^2\mathcal{P}_0(dw) \nonumber\\
& \quad + C\int||\hat{\pi}(z_i;\gamma_0)-\pi_0(z_i;\gamma_0)||^2\mathcal{P}_0(dw) \nonumber\\
& \quad + C\int||\hat{\pi}(z_i;\hat{\gamma})-\hat{\pi}(z_i;\gamma_0)||||\hat{\pi}(z_i;\gamma_0)-\pi_0(z_i;\gamma_0)||\mathcal{P}_0(dw)  \xrightarrow[]{p} 0
\end{align}

The first term converges in probability to 0 by Taylor series expansion.

Also by Assumption \ref{assmp-1} i) and ii), and
as just showed $||\hat{\pi}(z;\hat{\gamma}) -\pi_0(z;\gamma_0)|| \xrightarrow[]{p} 0$,    
    \begin{equation}
        \begin{split}
            \int ||\phi(w,\hat{\pi}(z;\hat{\gamma}),\alpha_0, \theta_0) - \phi(w,\pi_0,\alpha_0, \theta_0)||^2\mathcal{P}_0(dw) & = \int ||\alpha_0(z)(\pi_0(z;\gamma_0)-\hat{\pi}(z;\hat{\gamma}))||^2\mathcal{P}_0(dw) \\
            & \le C\int ||(\pi_0(z;\gamma_0)-\hat{\pi}(z;\hat{\gamma}))||^2\mathcal{P}_0(dw) \\
            & \le C||\hat{\pi}(z;\hat{\gamma})-\pi_0(z;\gamma_0)||^2 \xrightarrow[]{p} 0
        \end{split}
    \end{equation}

Also by Assumption \ref{assmp-1} i) and $||\hat{\alpha} -\alpha_0|| \xrightarrow[]{p} 0$, we have, 
\begin{equation}
        \begin{split}
            \int ||\phi(w,\pi_0(z;\gamma_0),\hat{\alpha}, \Tilde{\theta}) - \phi(w,\pi_0(z;\gamma_0),\alpha_0, \theta_0)||^2\mathcal{P}_0(dw) & = \int ||\left(\hat{\alpha}(z) -\alpha_0(z)\right)\left(y-\pi_0(z;\gamma_0)\right)||^2\mathcal{P}_0(dw) \\
            & \le C\int ||\hat{\alpha}-\alpha_0||^2\mathcal{P}_0(dw) \\
            & \le C||\hat{\alpha}-\alpha_0||^2 \xrightarrow[]{p} 0
        \end{split}
    \end{equation}

This satisfies Assumption 1 parts i), ii), and iii) of CEINR.

Next, consider 
$$
\begin{aligned}
\hat{\Delta}(w) &:=\phi\left(w, \hat{\pi}(z;\hat{\gamma}), \hat{\alpha}, \tilde{\theta}\right)-\phi\left(w, f_{0}, \hat{\alpha}, \tilde{\theta}\right)-\phi\left(w, \hat{\pi}(z;\hat{\gamma}), \alpha_{0}, \theta_{0}\right)+\phi\left(w, f_{0}, \alpha_{0}, \theta_{0}\right) \\
&=-\left[\hat{\alpha}(z)-\alpha_{0}(z)\right]\left[\hat{\pi}(z;\hat{\gamma})-\pi_0(z;\gamma)\right] .
\end{aligned}
$$
Then by the Cauchy-Schwartz inequality,
and Assumptions \ref{assmp-2} i) and ii)

\begin{equation}
\label{2i}
    \begin{split}
E\left[\hat{\Delta}(w)\right] &=\int-\left[\hat{\alpha}(z)-\alpha_{0}(z)\right]\left[(\hat{\pi}(z;\hat{\gamma})-\pi(z;\gamma))\right] \mathcal{P}_{0}(d z) \\
& \leq\left\|\hat{\alpha}-\alpha_{0}\right\|\left\|(\hat{\pi}(z;\hat{\gamma})-\pi(z;\gamma))\right\| =o_{p}\left(\frac{1}{\sqrt{n}}\right)        
    \end{split}
\end{equation}

Also since $\hat{\alpha}(z)$ and $\alpha(z)$ is bounded, we have
\begin{equation}
\label{2ii}
    \begin{split}
\int\left\|\hat{\Delta}(w)\right\|^{2} \mathcal{P}_{0}(d w) &=\int\left[\hat{\alpha}(z)-\alpha_{0}(z)\right]^{2}\left[(\hat{\pi}(z;\hat{\gamma})-\pi_0(z))\right]^{2} \mathcal{P}_{0}(d z) \\
& \leq C \int\left[(\hat{\pi}-\pi_0(z))\right]^{2} \mathcal{P}_{0}(d z) \stackrel{p}{\longrightarrow} 0  
    \end{split}
\end{equation}
Thus Equation \ref{2i} and Equation \ref{2ii} verify Assumption 2 i) of CEINR. 

Next Assumption 3 of CEINR is satisfied through Assumption \ref{assmp-3}. Thus Assumptions 1-3 of CEINR are satisfied. Thus asymptotic normality follows by Lemma 15 of CEINR and the Lindberg-Levy central limit theorem.

Finally, we know $\theta \stackrel{p}{\longrightarrow} \theta_0$. And thus we have, $\int\left\|g\left(w, \hat{\pi}(z;\hat{\gamma}), \tilde{\theta}\right)-g\left(w, \hat{\pi}(z;\hat{\gamma}), \theta_{0}\right)\right\|^{2} \mathcal{P}_{0}(d w) =  \stackrel{p}{\longrightarrow} 0$

To get the second conclusion, we need to show that $\hat{V}$ is a consistent estimator of $V$. To show this, we closely follow \cite{chernozhukov2021automatic}. Let $\psi_i=\psi_0\left(w_i\right)$ and consider
$$
\hat{V}=\frac{1}{n} \sum_{i=1}^n \hat{\psi}_i^2=\frac{1}{n} \sum_{i=1}^n\left(\hat{\psi}_i-\psi_i\right)^2+\frac{2}{n} \sum_{i=1}^n\left(\hat{\psi}_i-\psi_i\right) \psi_i+\frac{1}{n} \sum_{i=1}^n \psi_i^2
$$
hence, by re-arranging the terms and Cauchy-Schwarz inequality,
$$
\hat{V}-V=\frac{1}{n} \sum_{i=1}^n\left(\hat{\psi}_i-\psi_i\right)^2+\frac{2}{n} \sum_{i=1}^n\left(\hat{\psi}_i-\psi_i\right) \psi_i \leq \frac{1}{n} \sum_{i=1}^n\left(\hat{\psi}_i-\psi_i\right)^2+2 \sqrt{\frac{1}{n} \sum_{i=1}^n\left(\hat{\psi}_i-\psi_i\right)^2} \sqrt{\frac{1}{n} \sum_{i=1}^n \psi_i^2} .
$$
Using the triangle inequality, for $i \in I_{\ell}$,
$$
\left(\hat{\psi}_i-\psi_i\right)^2 \leq C \sum_{j=1}^4 R_{i j}=C \sum_{j=1}^3 R_{i j}+o_p(1)
$$

where
$$
\begin{aligned}
& R_{i 1}=\left[m\left(w_i, \hat{\pi}_{\ell}(z_i;\hat{\gamma}_{\ell})\right)-m\left(w_i, \pi_0(z_i;\gamma_{0})\right)\right]^2, \\
& R_{i 2}=\hat{\alpha}_{\ell}^2\left(z_i\right)\left[\hat{\pi}_{\ell}(z_i;\hat{\gamma}_{\ell})-\pi_0(z_i;\gamma_{0})\right]^2, \\
& R_{i 3}=\left[\hat{\alpha}_{\ell}\left(z_i\right)-\alpha_0\left(z_i\right)\right]^2\left[y_i-\pi_0\left(z_i; \gamma_0\right)\right]^2, \\
& R_{i 4}=\left(\hat{\theta}-\theta_0\right)^2 .
\end{aligned}
$$
We already showed consistency, so $R_{i 4} \stackrel{p}{\rightarrow} 0$. 

Let $I_{-\ell}$ denote observations not in $I_{\ell}$.
By Markov's inequality, for some $\delta>0$,
$$
\mathbb{P}\left(\frac{1}{n} \sum_{i=1}^n\left(\hat{\psi}_i-\psi_i\right)^2>\delta\right) \leq \frac{\mathbb{E}\left[\frac{1}{n} \sum_{i=1}^n\left(\hat{\psi}_i-\psi_i\right)^2\right]}{\delta}
$$
Note that the cross-fitting allows us to write
$$
\mathbb{E}\left[\frac{1}{n} \sum_{i=1}^n\left(\hat{\psi}_i-\psi_i\right)^2\right] \leq \mathbb{E}\left[\frac{C}{n} \sum_{\ell=1}^L \sum_{i \in I_{\ell}} \sum_{j=1}^3 R_{i j}\right]+o_p(1)=C \sum_{\ell=1}^L \frac{n_{\ell}}{n} \sum_{j=1}^3 \mathbb{E}\left[\mathbb{E}\left[R_{i j} \mid I_{-\ell}\right]\right]+o_p(1) .
$$

We already showed,
$$
\mathrm{E}\left[R_{i 1} \mid I_{-\ell}\right]=\int\left[m\left(w_i, \hat{\gamma}_{\ell}\right)-m\left(w_i, \gamma_0\right)\right]^2 F_0(d w)=o_p(1)
$$
Next by triangle inequality, we have 

\begin{align*} 
\mathrm{E}\left[R_{i 2} \mid \mathcal{W}_{-l}\right] ={} & \int\hat{\alpha}_{l}^2\left[\hat{\pi}_{l}(z_i;\hat{\gamma}_{l})-\pi_0(z_i;\gamma_{0})\right]^2 F_0(d z) \\
={}& \int\left[\hat{\alpha}_{l} + \alpha_0 - \alpha_0\right]^2\left[\hat{\pi}_{l}(z_i;\hat{\gamma}_{l})-\pi_0(z_i;\gamma_{0})\right]^2 F_0(d z) \\
\leq{} & \int\left[\hat{\alpha}_{l} - \alpha_0\right]^2\left[\hat{\pi}_{l}(z_i;\hat{\gamma}_{l})-\pi_0(z_i;\gamma_{0})\right]^2 F_0(d z) \\  
& + \int\left[ \alpha_0\right]^2\left[\hat{\pi}_{l}(z_i;\hat{\gamma}_{l})-\pi_0(z_i;\gamma_{0})\right]^2 F_0(d z) \\ 
\leq {} & O_p(1) \int\left[\hat{\pi}_{\ell}(z_i;\hat{\gamma}_{\ell})-\pi_0(z_i;\gamma_{0})\right]^2 F_0(d z) \stackrel{p}{\rightarrow} 0
\end{align*}

Finally, we have
$$
\begin{aligned}
\mathbb{E}\left[R_{i 3} \mid I_{-\ell}\right] & =\mathbb{E}\left[\mathbb{E}\left[\left[\hat{\alpha}_{\ell}\left(z_i\right)-\alpha_0\left(z_i\right)\right]^2\left[y_i-\pi_0\left(z_i;\gamma_0\right)\right]^2 \mid z_i, I_{-\ell}\right] \mid I_{-\ell}\right] \\
& =\mathbb{E}\left[\left[\hat{\alpha}_{\ell}\left(Z_i\right)-\alpha_0\left(Z_i\right)\right]^2 \mathbb{E}\left[\left[y_i-\pi_0\left(z_i;\gamma_0\right)\right]^2 \mid Z_i\right] \mid I_{-\ell}\right] \\
& \leq C\left\|\hat{\alpha}_{\ell}-\alpha_0\right\|^2 \stackrel{p}{\rightarrow} 0 .
\end{aligned}
$$
As a result,
$$
\frac{1}{n} \sum_{i=1}^n\left(\hat{\psi}_i-\psi_i\right)^2 \stackrel{p}{\rightarrow} 0
$$

Thus, we have $\hat{V} \stackrel{p}{\longrightarrow} V$

Also by Assumption 9 and iterated expectations
$$
\begin{aligned}
\mathrm{E}\left[R_{i 3} \mid \mathcal{W}_{-\ell}\right] & \leq \int\left\{\hat{\alpha}_{\ell}(z)-\bar{\alpha}(x)\right\}^2 \mathrm{E}\left[(y-\bar{\pi}(z))^2 \mid Z=z\right] F_Z(d z) \\
& \leq C \int\left\{\hat{\alpha}_{\ell}(z)-\bar{\alpha}(z)\right\}^2 F_z(d z)=C\left\|\hat{\alpha}_{\ell}-\bar{\alpha}\right\|^2=o_p(1) .
\end{aligned}
$$

\end{proof}

\section{Hyperparameter Space for Tuning Non-parametric Estimator Benchmark}\label{appendix_hyper}

\begin{table}[H]
    \centering
    \begin{tabular}{ll}
    \hline
    Hyperparameter & Space \\
    \hline
        Number of hidden layers & [3, 4, 5] \\
        Number of nodes in each layer &  [64, 128, 256] \\
        Learning rate & [1e-2, 1e-3, 1e-4] \\
        Number of epochs &  [1, 2, 4]\\
        \hline
    \end{tabular}
    \caption{Hyperparameter Space for Tuning Non-parametric Estimator Benchmark}
    \label{tab:my_label}
\end{table}

\section{Distribution of Features and Coefficients in Numerical Experiments}\label{sec:appendix_dist}
\begin{table}[h]
    \centering
    \begin{tabular}{lcc}
    \hline
        & MNL & RCL \\
    \hline
         $p_{jt}$ &   $U[0,4]$  & $U[0,4]$     \\
         $X_{jt}$ &   $N(0,1)$  & $N(0,1)$  \\
    \hline
    \end{tabular}
    \caption{Distribution of Features}
    \label{tab:feature}
\end{table}

\begin{table}[h]
        \centering
        \begin{tabular}{lcc}
        \hline
            & MNL & RCL \\
        \hline
             $\alpha_i$ &   -1  & $N(-1,1)$     \\
             $\beta_{ik}$ &   1  & $N(\mu_{\beta_k},1)$   \\
        \hline
        \end{tabular}
        \par
\smallskip 
\footnotesize 
\noindent
    Notes : $  \mu_{\beta_k} \sim N(0,1/2d)$ 
        \caption{Distribution of Coefficients}
        \label{tab:coef}
    \end{table}

\section{An Numerical Experiment of Endogeneity} \label{sec:appendix_endo}

In $\S$\ref{ssec:endo}, we show how our model handles endogeneity in theory. Here, we provide a simulation that demonstrates the performance of our model handling endogeneity. 

Let the utility $u_{ijt}$ that consumer $i$ in market $t$ derives from product $j$ as the following linear function
\begin{equation} \label{eq:dgp_endo1}
    u_{ijt} = \alpha_i p_{jt} + \beta_i X_{jt} + \gamma_i \mu_{jt}+  \varepsilon_{ijt} ,
\end{equation}
where $\mu_{jt}$ is the unobservable that correlates $p_{jt}$ and $\varepsilon_{ijt}$ is i.i.d. Type-I extreme value distributed.  Specifically, without loss of generality, we specify the correlation between $p_{jt}$ and $\mu_{jt}$ as 
\begin{equation}\label{eq:dgp_endo2}
    p_{jt} = \delta_1 IV_{jt} +  \delta_2 \mu_{jt} ,
\end{equation}
where $IV_{jt}$ is the exogenous instrumental variable.

In this simulation, we first generate $X_{jt}$ and $\mu_{jt}$, following the same distribution as the $X_{jt}$ in our baseline model, and $IV_{jt}$, following the same distribution as our as $p_{jt}$ in our baseline model, as discussed in Web Appendix \ref{sec:appendix_dist}. We next generate $p_{jt}$ using Equation \ref{eq:dgp_endo2}. Finally, we generate market shares assuming that the true model is RCL. For simplicity, we let $\delta_1 = \delta_2 = 1$. In addition, $\alpha_i$ and $\beta_i$ follow the same distribution as the baseline simulations and $\gamma_i$ follows the same distribution as $\beta_i$, as stated in Web Appendix \ref{sec:appendix_dist}. We consider a case with 10 products, 100 markets, and 10 non-price features (in addition to price), the same as our baseline simulation in $\S$\ref{sssec:mnl_rcl}. In the estimation, we let $\mu_{jt}$ be an unobservable. We use the OLS linear regression as our first-stage regressor. 

We consider two cases for comparison: 
\begin{itemize}
    \item \textbf{Exogeneous Benchmark}: Uses the exact same DGP but treat  $\mu_{jt}$ as observables to researchers in estimation. Since the coefficients of price and market shares are the same as the endogeneity case, true elasticities are the same for these two DGPs. This gives a benchmark performance with the same data under the assumption that endogeneity is not present.
    \item \textbf{Ignoring Endogeneity}: Directly trains the model using only the observed features  $X_{jt}$ and $p_{jt}$ without considering the endogeneity problem. Note that, even though the endogeneity problem is ignored, this does not mean the predictive performance in market shares would be low for this case because the model could be overfitted. However, the elasticities will be biased when the endogeneity is ignored. 
\end{itemize}

Following the routine in our main text, we simulate 20 times for the same DGP with different parameters and features and report the predictive performance on market share  ($\hat{\pi}_{jt}$), own price elasticity ($\frac{\partial \hat{\pi}_{jt} / \pi_{jt}  }{\partial p_{jt}/p_{jt}}$), and cross-price elasticity ($\frac{\partial \hat{\pi}_{jt} / \pi_{jt}  }{\partial p_{k\neq j t}/  p_{k\neq j t} }$). We report the performance of our model in Table \ref{tbl:endo}. When we use the control function to correct for endogeneity (Row 1: Our Method), our model underperforms slightly compared to the case where $\mu_{jt}$ is treated as an observable (Row 2: Exogenous Benchmark), on all three metrics. However, when we ignore the endogeneity issue and apply our model (Row 3: Ignore Endogneity), although the predictive performance of market shares is not bad, the estimation of own- and cross-elasticities are significantly biased, with the MAEs being 10--25 times higher than the cases where we account for endogeneity. This demonstrates both the importance of accounting for endogeneity as well as the effectiveness of our method in handling this problem in real settings.

\begin{table}[htbp!]
\centering
\begin{tabular}{llllllll}
\hline\hline
                     & \multicolumn{2}{c}{Market shares} & \multicolumn{2}{c}{Own-elasticity} & \multicolumn{2}{c}{Cross-elasticity} \\ 
                     & MAE            & RMSE             & MAE            & RMSE              & MAE             & RMSE               \\ \hline
Our Method         & 0.0256         & 0.0194           & 0.2307         & 0.3516            & 0.0669          & 0.1033             \\
Exogeneous Benchmark          & 0.0254         & 0.0195           & 0.2290         & 0.3469            & 0.0657          & 0.1039             \\
Ignore Endogeneity  & 0.0263         & 0.0202           & 2.3832         & 5.1901            & 1.9751          & 4.4306             \\ \hline\hline
\end{tabular}
\caption{Model Performance of Endogeneity Case}
\label{tbl:endo}
\end{table}

\section{Details in Adopting the ``MLIV” Method}\label{appendix_mliv}
Following \citet{singh2020machine}, we perform the steps below to construct the machine-learning-based IV (MLIV) and use them to estimate $\hat{\gamma}$ to control for endogeneity in prices.

\begin{itemize}
    \item \textbf{Step 1: Data Partition} We randomly split the data set, $\mathcal{D}$, into three separate partitions of markets, each denoted as $D_l$. Each market is exclusively assigned to only one partition. For each partition, we define its complement set, $D_l^c$, as the subset of data in $\mathcal{D}$ that is not included in $D_l$.
    \item \textbf{Step 2: Cross-fitting} For each partition $l$,  we first estimate a linear regression model on the complement data set, $D_l^c$, using the Lasso method with hyperparameters tuned by 3-fold cross-validation. As discussed in section \ref{ssec:inference}, we need the estimator of $\gamma$ to converge at $n^{-1/2}$ rate, a similar result that bounds the prediction error of the lasso estimator has been established in \cite{chatterjee2015prediction}. Then, we use this trained model to predict the outcomes (prices) of the $D_l$. We denote the fitted value as $\hat{f}_l$, which is essentially the MLIV. 
    \item \textbf{Step 3: First-stage Regression} We estimate the first-stage estimator $\gamma$ and residual $\mu$ using the MLIV as the only predictor following step 1 in Section \ref{sec:estimation}. 
\end{itemize}

As a supplement to our main result, we also run our model using non-machine learning-based IVs. Similar to Figure 4 in the main text, we present the estimated own-elasticity of our model without IV, with BLP-style IVs, with differentiation IVs, and with MLIV in Figure \ref{fig:IV_diff_apd}. In Figure \ref{fig:rawIV_apd}, even when IVs are applied, the persistence of many positive own-elasticities suggests the weakness of the BLP style IVs. Furthermore, we apply the differentiation IVs \parencite{gandhi2019measuring}, which use exogenous measures of differentiation and provide a more robust instrument compared to the conventional BLP IVs. As one can see from  Figure \ref{fig:diffIV_apd}, the use of differentiation IV provides a more realistic estimation of own-elasticities, strengthening the issue of weak instruments of the BLP style IVs. We also include the distributions of the estimated own- and cross-elasticities obtained from our model using different sets of IVs in Figure \ref{fig:blp_elasticity_apd}. 

\begin{figure}[h]
\centering
\begin{subfigure}{0.45\textwidth} 
\includegraphics[width=\textwidth]{figures/elas_noIV.png}
\caption{Without IV}
\label{fig:noIV_apd}
\end{subfigure}%
\hfill 
\begin{subfigure}{0.45\textwidth} 
\includegraphics[width=\textwidth]{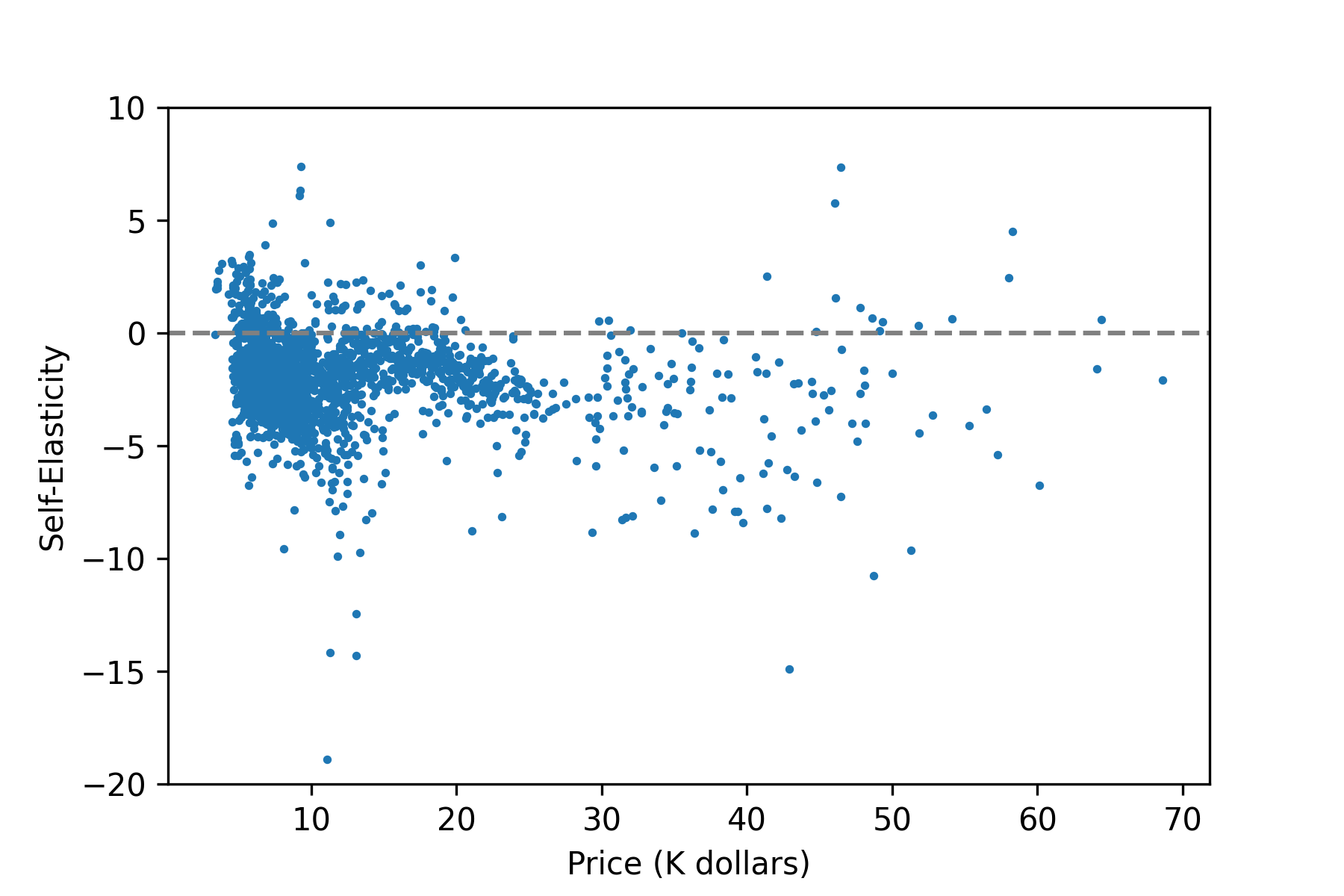}
\caption{BLP Style IVs}
\label{fig:rawIV_apd}
\end{subfigure}
\\ 
\begin{subfigure}{0.45\textwidth} 
\includegraphics[width=\textwidth]{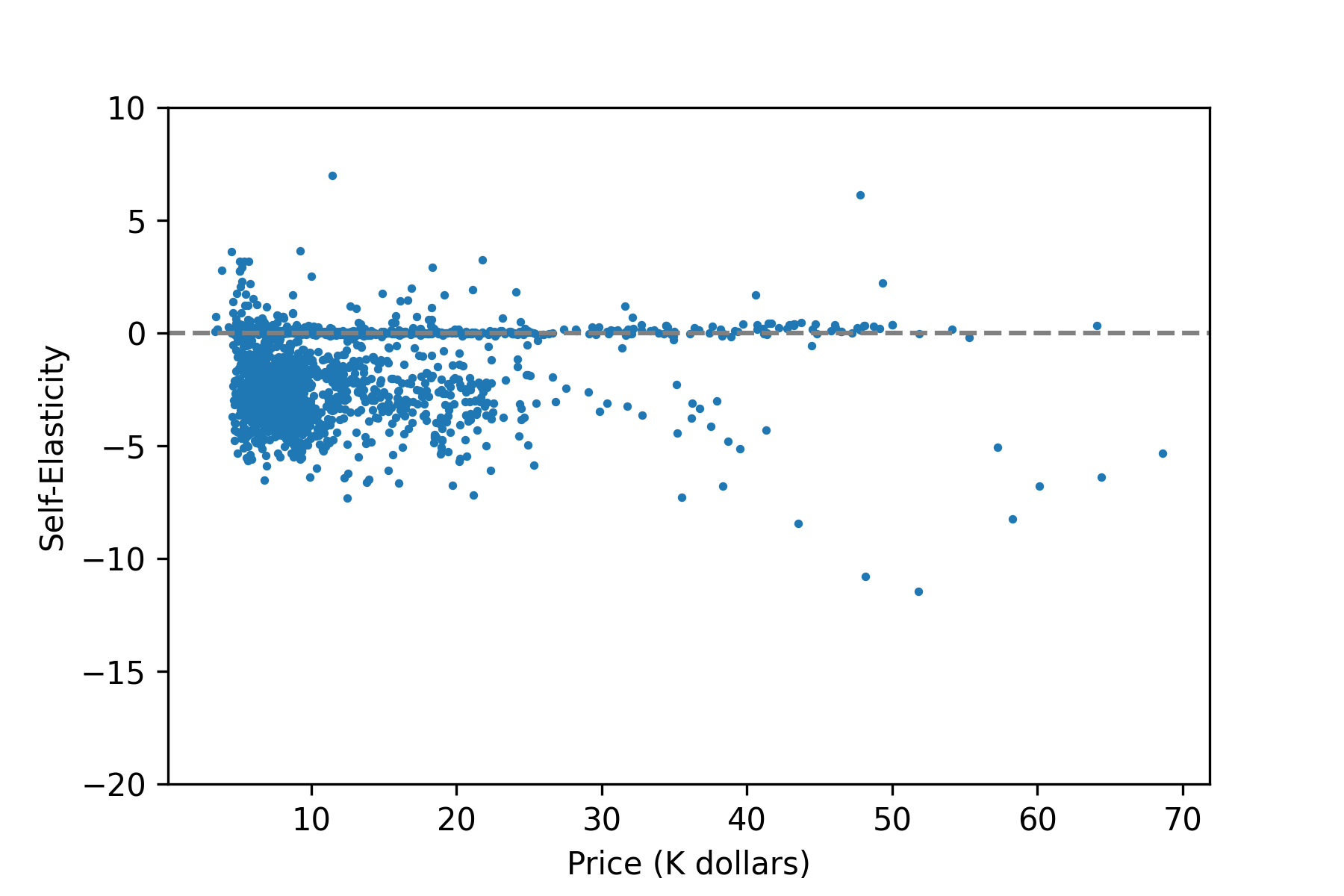}
\caption{Differentiation IVs}
\label{fig:diffIV_apd}
\end{subfigure}%
\hfill 
\begin{subfigure}{0.45\textwidth} 
\includegraphics[width=\textwidth]{figures/elas_IV.png}
\caption{MLIV}
\label{fig:withIV_apd}
\end{subfigure}
\caption{Elasticity Estimation Comparison}
\label{fig:IV_diff_apd}

\footnotesize{Figure \ref{fig:IV_diff_apd} presents the estimated own-elasticity of our model without IV, with BLP Style IVs, with differentiation IVs and with MLIV. The x-axis represents the price of the focal product, while the y-axis shows the product's own-elasticity. Each point corresponds to a product in a market, resulting in 2,217 observations. We report the estimated elasticity based on the same price variation used in the BLP paper (a 1,000-dollar change). }
\end{figure}

\begin{figure}[htbp]
\centering
\begin{subfigure}{\textwidth}
    \centering
    \includegraphics[width=0.7\textwidth]{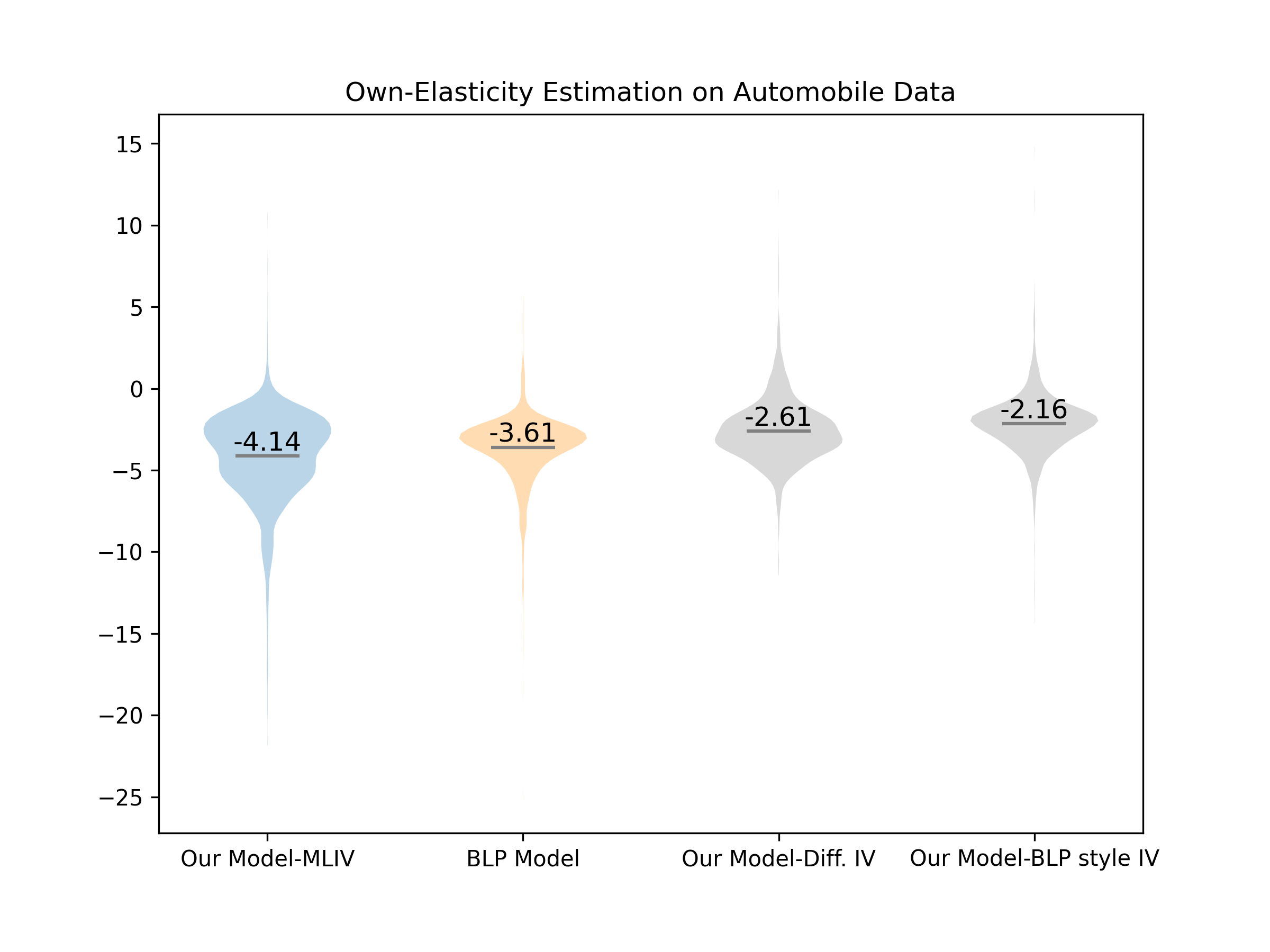}
    \caption{Own-Elasticity Estimation (Our Model vs. BLP Model)}
    \label{fig:blp_own_ela_apds}
\end{subfigure}
\\ 
\begin{subfigure}{\textwidth}
    \centering
    \includegraphics[width=0.7\textwidth]{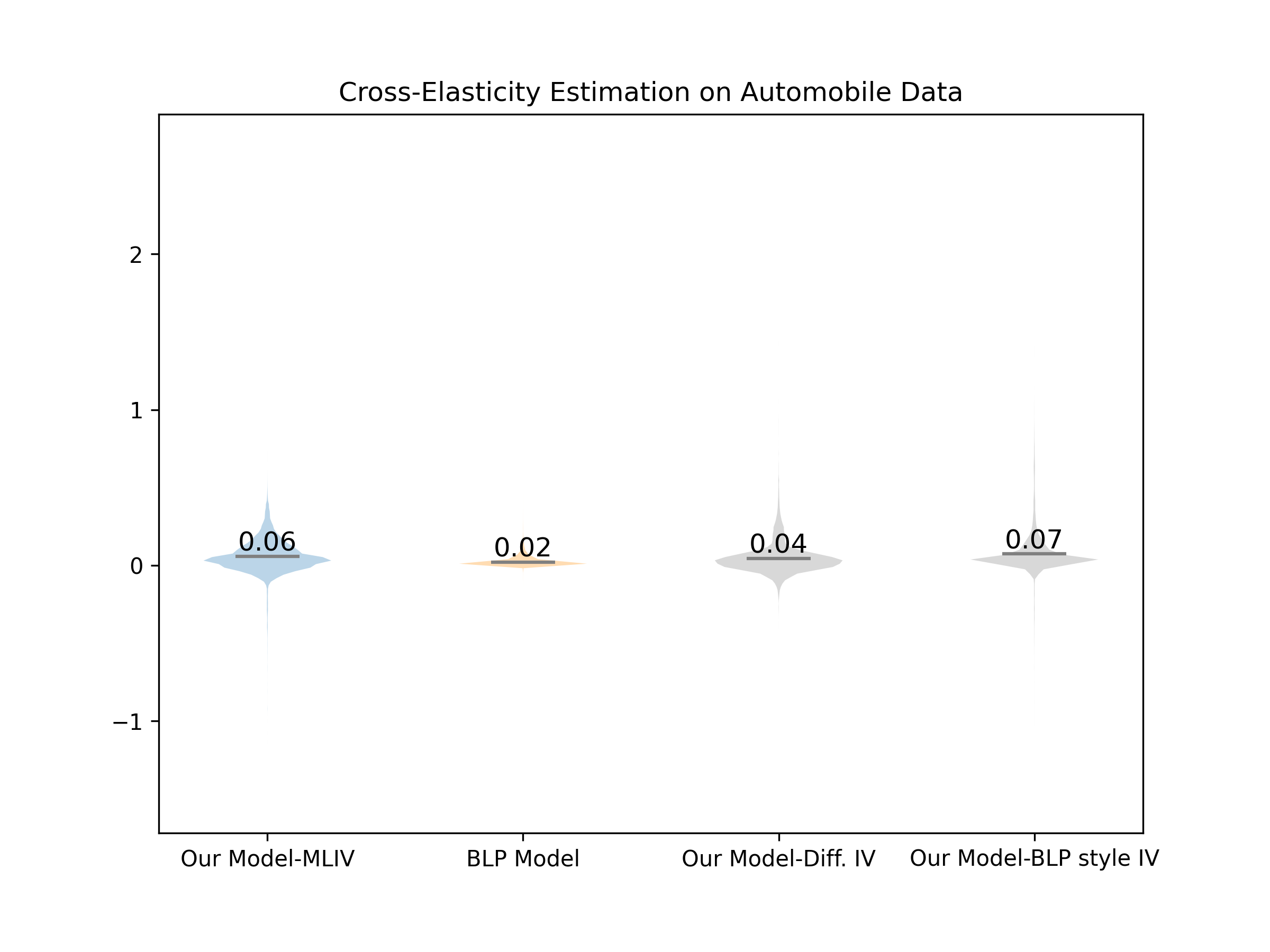}
    \caption{Cross-Elasticity Estimation (Our Model vs. BLP Model)}
    \label{fig:blp_cross_elas_apd}
\end{subfigure}
\caption{Elasticity Estimation Comparison}
\label{fig:blp_elasticity_apd}

\footnotesize{Note: Figure \ref{fig:blp_elasticity_apd} illustrates the distributions of the estimated own- and cross-elasticities obtained from our model (using different sets of IVs) and the BLP model. The filled areas in the violin plots represent the complete range of the elasticities, while the text labels indicate the mean values. }

\end{figure}

In addition, we also perform a weak instrument test on both BLP Style IVs and the MLIV and report the F-statitics and p-value in Table \ref{tab:iv_test}. Both BLP Style IVs and MLIV pass the weak instrument tests. 

\begin{table}[H]
    \centering
    \begin{tabular}{lcc}
    \hline\hline
         &  F-statistic & P-value \\
    \hline
     BLP Style IVs    & 241.5 & $<$ 1e-8  \\
    MLIV    & 280.9 & $<$1e-8  \\
    \hline\hline
    \end{tabular}   
    \caption{Weak Instrument Test}
    \label{tab:iv_test}
\end{table}

\end{appendices}

\end{document}

%% file: header.tex
\usepackage{amsfonts}

\def\*{\star}

\DeclareMathOperator*{\argmin}{arg\,min}

\usepackage{graphicx}
\usepackage{mathtools}
\usepackage{footnote}
\usepackage{float}
\usepackage{xspace}
\usepackage{multirow}
\usepackage{wrapfig}
\usepackage{framed}
\usepackage{xcolor}
